\documentclass[11pt,reqno]{article}   
\usepackage{fullpage}

\usepackage[unicode=true]{hyperref}
\usepackage{amsmath}
\usepackage{cite}
\usepackage{amsfonts}
\usepackage{amssymb}
\usepackage{amsthm}
\usepackage{authblk}
\usepackage{mypack} 
\usepackage[pdftex]{color,graphicx}
\usepackage{adjustbox}
\usepackage{soul}
\usepackage{comment}
\usepackage{bbold}
\usepackage{tikz} 
\usetikzlibrary{decorations.pathreplacing,angles,quotes}
\usetikzlibrary{matrix,arrows,decorations.pathmorphing}
\usepackage{tikz-cd}
\usetikzlibrary{arrows}
\usetikzlibrary{decorations.markings}

\usepackage{bbold}
\usepackage{diagbox}
\usepackage{caption}
\usepackage{subcaption}
\usepackage{pdfpages} 
\usepackage{blkarray}
\usepackage{centernot}
\usepackage{mathtools}
\usepackage{stmaryrd}
\usepackage{soul}  

\definecolor{bluegray}{rgb}{0.4, 0.6, 0.8}
\definecolor{turquoise}{rgb}{0.2, 0.7, 0.6}

\title{Topological methods for studying contextuality:\\
{\Large $N$-cycle scenarios and beyond}
}

\author{Aziz Kharoof\footnote{aziz.kharoof@bilkent.edu.tr}}
\author{Selman Ipek\footnote{selman.ipek@bilkent.edu.tr}}
\author{Cihan Okay\footnote{cihan.okay@bilkent.edu.tr}}
\affil{{\small{Department of Mathematics, Bilkent University, Ankara, Turkey}}}

\date{\today}

\begin{document}
  \maketitle  
  
\begin{abstract}
{
Simplicial distributions are combinatorial models describing distributions on spaces of measurements and outcomes that generalize non-signaling distributions on contextuality scenarios.
This paper studies simplicial distributions on $2$-dimensional measurement spaces by introducing new topological methods.
Two key ingredients are a geometric interpretation of Fourier--Motzkin elimination and a technique based on collapsing of measurement spaces. 
Using the first one, we provide a new proof of Fine's theorem characterizing non-contextual distributions on $N$-cycle scenarios. 
Our approach goes beyond these scenarios and can describe non-contextual distributions on scenarios obtained by gluing cycle scenarios of various sizes. 
The second technique is used for detecting contextual vertices and deriving new Bell inequalities. Combined with these methods, we explore a monoid structure on simplicial distributions. 
}
\end{abstract}

\tableofcontents

\section{Introduction}

{Quantum contextuality is a fundamental feature of collections of probability distributions obtained from quantum measurements.}
{In a classical setting, experimental statistics are derivable from a joint probability distribution. Measurements of quantum observables, however, do not satisfy this principle, leading to violations of Bell inequalities, or more generally, non-contextual inequalities, which serve as a witness of this quintessentially non-classical phenomenon. That such violations were \textit{necessary} was first discovered by Bell \cite{bell64}. Later, Fine \cite{fine1982hidden,fine1982joint} showed that such inequalities were also \textit{sufficient} for recovering a classical description in the well-known Clauser--Horne---Shimony--Holt (CHSH) scenario \cite{chsh69}.}

A systematic study of contextuality scenarios using sheaf theory 
{was}
introduced by Abramsky--Brandenburger in \cite{abramsky2011sheaf}.
Later topological ideas from group cohomology 
{were}
introduced to the study of contextuality \cite{Coho}, {with an emphasis on} investigating quantum advantage in measurement-based quantum computation. {More recently, a unified framework for the study of contextuality was introduced, based on combinatorial representations of topological spaces known as simplicial sets \cite{okay2022simplicial}.}
The basic objects in this theory are called simplicial distributions. {This theory} subsumes the theory of non-signaling distributions and goes beyond by formulating the notion of distributions on spaces rather than sets.
{Contextuality can be formulated in this generality.}
 
\begin{figure}[h!]
\centering
\includegraphics[width=.4\linewidth]{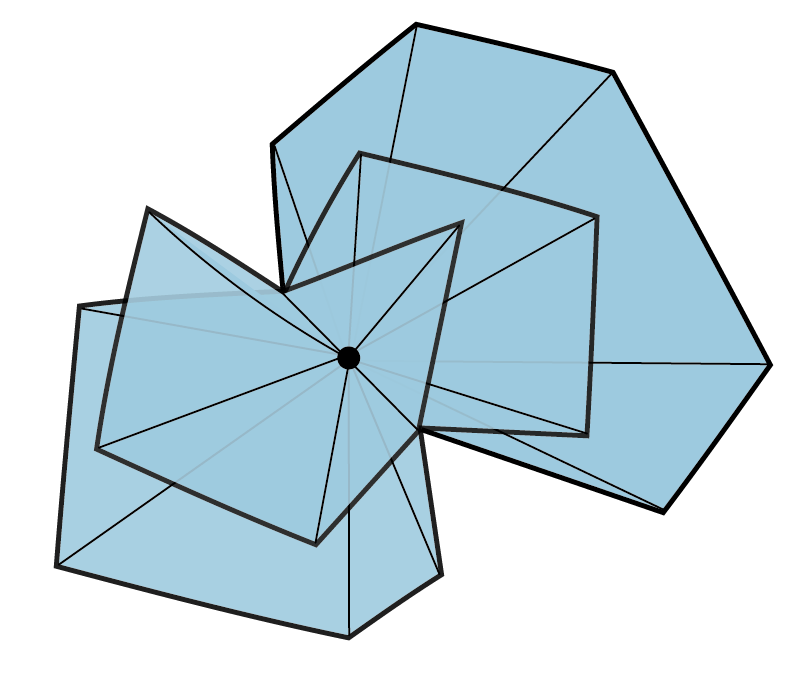}
\caption{{Flower scenario.}
}
\label{fig:flower}
\end{figure}  
 
{Initial applications of simplicial distributions in \cite{okay2022simplicial} included a new topological proof of Fine's theorem for the CHSH scenario. A novel feature of this approach is its flexibility in realizing measurement scenarios as topological spaces. Such expressiveness allows for contextuality to be characterized topologically in multiple ways. For instance, one realization of the CHSH scenario is topologically equivalent to a disk consisting of four triangles, while another realization, also appearing in \cite{okay2022mermin}, is given by a punctured torus. While the former allows for an analysis similar in spirit to that of Fine, the latter work supplies an alternative proof by the classifying the extreme distributions on the torus. }%
In this paper, we go beyond these examples and consider a generalization of $N$-cycle scenarios
{\cite{araujo2013all,braunstein1990wringing}} 
 which we call {\it flower scenarios}.
The flower scenario is obtained by gluing various cycle scenarios of arbitrary size as in Fig.~(\ref{fig:flower}).
{
This scenario is a particular example of a class of 
$2$-dimensional measurement spaces.
Given  
a $1$-dimensional simplicial set, i.e., a graph, 
{the cone construction produces a $2$-dimensional simplicial set.}
This construction introduces a new vertex and a set of
triangles connecting each edge on the graph to the new vertex.
For a $1$-dimensional space $X$, we will write $\Cone(X)$ for the cone space. We {will} write $L_N$ for the space obtained by gluing $N$ edges in the shape on a line.
}

\begin{customthm} {\ref{thm:bouquet-of-lines}}
{Let $\Cone(X)$  denote the flower scenario  (Fig.~(\ref{fig:flower})), {the cone of $X$ obtained by gluing the lines $L_{N_1},\cdots,L_{N_k}$ at their end points}.}
A simplicial distribution $p\in \sDist(\Cone(X))$ is non-contextual if and only if for every $N$-{circle}   $C$ on $X$ the restriction {$p|_C$} satisfies the $N$-circle inequalities\footnote{{In the literature {what we refer to as} $N$-circle inequalities are known as $N$-cycle inequalities. We diverge in terminology by emphasizing the underlying topological space, which is a circle.}}.
\end{customthm}

The primary technique that goes into the proof of this result is the Fourier--Motzkin (FM) elimination \cite{chvatal1983linear}, a version of Gaussian elimination for inequalities.
In Section \ref{sec:FM} we present a topological interpretation of FM elimination.
A measurement space is represented by a simplicial set whose simplices correspond to measurements. In this paper, we will restrict our attention to $2$-dimensional simplicial sets, that is, those obtained by gluing triangles. Our outcome space will be fixed to a canonical choice obtained from $\ZZ_2=\set{0,1}$ (known as the nerve space) so that the measurements labeling the edges have binary outcomes. 
In this setting, non-contextuality is characterized by Bell inequalities consisting of variables corresponding to probabilities of {measurements on the edges.} 
For our topological proof of Fine's theorem we consider  {a particular} triangulation of the disk, which we refer to as a {\it classical $N$-disk}. On these disks any simplicial distribution turns out to be non-contextual, hence the name classical.
If we start from a distribution on the boundary of a disk, the $N$-{circle inequalities} appear as the sufficient and necessary condition for extending such a distribution from the boundary to the entire disk (Proposition \ref{pro:n-cycle-ineq-FM}).
Now, given two such classical disks glued at a common edge, the topological interpretation of FM elimination is that the boundary of the new space is formed by taking the union of the boundaries of the disks and omitting the common edge; {see Fig.~(\ref{fig:bouquet})}. 
The elimination of the edge is the geometric interpretation of removing the variable by FM elimination.
This key idea allows us to characterize the extension condition from the boundary of a bouquet of classical disks, i.e., a collection of disks glued at a common edge, by a collection of {circle inequalities} (Corollary \ref{cor:generalized-extension}). This extension result is the main ingredient of the proof of Theorem \ref{thm:bouquet-of-lines} that characterizes non-contextual distributions on the flower scenario {Fig.~(\ref{fig:flower})}. Note that this scenario generalizes bipartite Bell scenarios where Alice performs $2$ measurements, and Bob performs $m$ measurements, {and} all measurements have binary outcomes {\cite{collins2004relevant}}. 

Our next main contribution is the collapsing of measurement spaces to detect contextual vertices of simplicial distributions (Section \ref{sec:collapsing-measurement-spaces}).
{To study simplicial distributions on the cone space we introduce a technique based on collapsing edges. Let $\pi:X\to X/\sigma$ denote the map that collapses an edge $\sigma$ of the graph.
Applying the cone construction gives a map $\Cone\pi:  {\Cone(X)\to \Cone(X/\sigma)} $ between the cone spaces.
A simplicial distribution on the cone of the collapsed measurement space can be extended via $\Cone\pi$ to give a simplicial distribution on the cone of the original measurement space. We denote this map by 
$${(\Cone\pi)^*:\sDist(\Cone(X/\sigma))\to \sDist(\Cone (X))}.$$ In Theorem \ref{thm:collapsing} we show that for a simplicial distribution $p\in \sDist(\Cone(X/\sigma))$ and its image  $q=(\Cone \pi)^*(p)$ the following holds:
\begin{enumerate} 
\item $p$ is a vertex if and only if $q$ is a vertex.
\item $p$ is contextual if and only if $q$ is contextual.
\item $p$ is strongly contextual  if and only if $q$ is strongly contextual.
\item $p$ is a deterministic distribution if and only if $q$ is a deterministic distribution.
\end{enumerate} }
\noindent
In particular, parts (1) and (2) imply that contextual vertices map to contextual vertices under the collapsing map.
This method is very powerful in detecting vertices.
Let $n_X$ denote the number of generators of the fundamental group of the graph $X$.
Then the number of contextual vertices in $\sDist(\Cone(X))$ is lower bounded by $(2^{n_X}-1)2^{|X_0|-1}$ where $|X_0|$ denotes the number of vertices of the graph (Theorem \ref{thm:GenrVert}). 
In addition, we use this method to derive new Bell inequalities from known ones. For example, the Froissart inequalities \cite{froissart1981constructive} of the scenario given by the cone of the bipartite graph $K_{3,3}$ produce new Bell inequalities for the {cone of the} graph obtained by collapsing one of the edges (Section \ref{sec:application-vell-ineq}).

Finally, we explore a new algebraic feature of simplicial distributions first introduced in \cite{kharoof2022simplicial}.
{T}he set of simplicial distributions $\sDist(X)$ has a monoid structure. Together with its polytope structure, this gives a convex monoid.
The restriction of the monoid structure to deterministic distributions gives a group structure, and this group acts on  simplicial distributions. 
Using this action, we can generate more vertices from those obtained from the collapsing technique.
Our other contributions are as follows: 
(1) For a $2$-dimensional measurement space $X$ 
we show that $\sDist(X)$ is a convex polytope (Proposition \ref{pro:fConvex}) and provide the $H$-description {(Corollary \ref{cor:H-description})}. 
(2) We describe the monoid structure on $\sDist(X)$ (Section \ref{sec:prod}) and describe the action of the set of deterministic distributions on Bell inequalities and contextual vertices (Example \ref{ex:chsh}).
(3) The $1$-cycle scenario obtained as the cone of a circle {(Fig.~(\ref{fig:D-S-vert1}))} is a new scenario that cannot be realized in the conventional non-signaling picture. {More generally,} we describe the polytope of simplicial distributions on the cone of the wedge $\vee_{i=1}^n C_1$ of $1$-{circle}s (Proposition \ref{pro:ConeofCircCube}).

\section{Simplicial distributions}

The theory of simplicial distributions is introduced in \cite{okay2022simplicial}. 
A simplicial distribution is defined for a space of measurements and outcomes.
In this formalism spaces are represented by combinatorial objects known as simplicial sets.
More formally, a {\it simplicial set} $X$ consists of a sequence of sets $X_n$ for $n\geq 0$ and the simplicial structure maps: 
\begin{itemize}
\item Face maps $d_i:X_n\to X_{n-1}$ for $0\leq i \leq n$.
\item Degeneracy maps $s_j:X_n\to X_{n+1}$ for $0\leq j\leq n$.
\end{itemize} 
These maps are subject to simplicial relations (see, e.g., \cite{friedman2008elementary}).
An $n$-simplex is called  {\it degenerate} if it  lies in the image of a degeneracy map, otherwise it is called  {\it non-degenerate}. Geometrically only the non-degenerate simplices are relevant.
Among the non-degenerate simplices there are ones that are not a face of another non-degenerate simplex. 
Those simplices we will refer to as {\it generating simplices}. 
{Throughout the paper when we refer to an edge ($1$-simplex) or a triangle ($2$-simplex) of a simplicial set we mean a non-degenerate one.}

In this paper we will focus on spaces obtained by gluing triangles.

\Ex{\label{ex:triangle} 
{\rm
The triangle, denoted by $\Delta^2$, is the simplicial set with simplices
$$
(\Delta^2)_n = \set{\sigma^{a_0\cdots a_n}: \, 0\leq a_0\leq \cdots a_n \leq 2,\, a_i\in \ZZ  }.
$$ 
The $i$-th face map deletes the $i$-th index: $d_i(\sigma^{a_0\cdots a_n}) = \sigma^{a_0\cdots a_{i-1}a_{i+1}\cdots a_n}$, and the $j$-th degeneracy map copies the $j$-th index: $s_j(\sigma^{a_0\cdots a_n}) = \sigma^{a_0\cdots a_ja_j\cdots a_n}$. The simplex $\sigma^{012}$ is 
the generating simplex. Any other simplex can be obtained by  applying   a sequence of face and degeneracy maps. {In general, we can define  $\Delta^d$ consisting of $n$-simplices of the form $\sigma^{a_0\cdots a_n}$ where $0\leq a_0\leq \cdots \leq a_n\leq d$. This simplicial set represents the topological $d$-simplex. Of particular interest are $\Delta^0$ and $\Delta^1$  representing a point and an edge, respectively.
}
}
}

The gluing operation can be specified by introducing relations between the generating simplices. The simplest example is obtained by gluing two triangles along a face.

\Ex{\label{ex:diamond} 
{\rm
The diamond space $D$ is defined as follows:
\begin{itemize}
\item Generating simplices: $\sigma^{012}_A$ and $\sigma_B^{012}$.
\item Identifying relation: $d_1\sigma^{012}_A = d_1 \sigma^{012}_B$.
\end{itemize}
We can define other versions by changing the faces. We will write $D_{ij}$ for the diamond whose identifying relation is $d_i\sigma^{012}_A = d_j \sigma^{012}_B$.
}
}

Next we introduce the notion of maps between simplicial sets.
A {\it map $f:X\to Y$ of simplicial sets} consists of a sequence $f_n:X_n\to Y_n$ of functions that respect the face and the degeneracy maps.
Given a simplex $\sigma\in X_n$ we will write $f_\sigma$ for $f_n(\sigma)\in Y_n$. With this notation the compatibility conditions   are given by
$$
d_i f_\sigma = f_{d_i\sigma} \;\; \text{ and }\;\; s_j f_\sigma = f_{s_j\sigma}.
$$
A simplicial set map $f:\Delta^2\to Y$ is determined by the image of the generating simplex, that is, by an arbitrary $2$-simplex $f_{\sigma^{012}}\in Y_2$. Therefore these maps are in bijective correspondence with the elements of $Y_2$. 
In the case of the diamond space a simplicial set map $f:D_{ij}\to Y$ is determined by  $f_{\sigma_A^{012}}$ and $f_{\sigma_B^{012}}$ satisfying 
$$
d_if_{\sigma_A^{012}} =  f_{d_i\sigma_A^{012}} = f_{d_j\sigma_B^{012}} = d_j f_{\sigma_B^{012}}.
$$

Given a {simplicial set}
$Y$ we will construct another simplicial set that represents the space of distributions on $Y$. For this we need the distribution monad $D_R$ defined for a commutative semiring $R$ \cite{jacobs2010convexity}. 
Throughout the paper we will take $R$ to be $\RR_{\geq 0}$. 
A distribution on a set $U$ is defined to be a function $p:U\to R$ of finite support such that $\sum_{u\in U} p(u)=1$.
{
The delta distribution at $u\in U$ is defined by
$$
\delta^u(u') = \left\lbrace
\begin{array}{ll}
1 & u'=u \\
0 & \text{otherwise.}
\end{array}
\right.
$$
Any distribution can be expressed as a sum of delta distributions: $p=\sum_{u \in U}p(u)\delta^u$.
}
For a function $f:U\to V$ we will write $D_Rf$ for the function $D_R(U)\to D_R(V)$ defined by
$$
D_Rf(p)(v) = \sum_{u\in f^{-1}(v)} p(u).
$$
{The space of distributions on $Y$} is represented by the simplicial set $D_R(Y)$ whose $n$-simplices are given by $D_R(Y_n)$. The face and the degeneracy maps are given by $D_Rd_i$ and $D_R s_j$.
There is a canonical simplicial set map $\delta:Y\to D_R(Y)$ defined by sending a simplex $\sigma$ to the delta distribution {$\delta^\sigma$}.

\Def{\label{def:simplicial-distribution}
{A {\it simplicial scenario} consists of a pair $(X,Y)$ of simplicial sets where $X$ represents the space of measurements and $Y$ represents the space of outcomes.}
A {\it simplicial distribution} {on $(X,Y)$} is a simplicial set map $p:X\to D_R(Y)$.
A  simplicial set map of the form $s:X\to Y$ 
is called an {\it outcome assignment}. 
The associated distribution $\delta^s:X\to D_R(Y)$ is defined to be the composite $\delta \circ s$ is called a {\it deterministic distribution}.
We will write $\sDist(X,Y)$ and $\dDist(X,Y)$ for the set of simplicial and deterministic distributions.
}

There is a canonical map
$$
\Theta: D_R(\dDist(X,Y))\to \sDist(X,Y)
$$
defined by sending $d=\sum_{s} d(s) \delta^s$ to the simplicial distribution $\Theta(p)$ defined by
$$
\Theta(p)_\sigma = \sum_{s} d(s) \delta^{s_\sigma}.
$$

\Def{\label{def:contextual}
A simplicial distribution $p:X\to D_RY$ is called {\it non-contextual} if $p$ is in the image of $\Theta$. Otherwise, it is called {\it contextual}.
}

There is a stronger version of contextuality whose definition relies on the notion of support.
The  {\it support} of a simplicial distribution $p:X\to D_RY$  is defined by
\begin{equation}\label{eq:support-simp-dist}
\supp(p)=\set{s:X\to Y\,:\, p_{\sigma}(s_{\sigma})>0\;\forall \sigma\in X_n,\;n\geq 0}.
\end{equation}

\Def{\label{def:StronglyContextual}
A simplicial distribution $p$ on   $(X,Y)$ is called  {\it strongly contextual} if its  support $\supp(p)$ is empty.
}

\subsection{Two-dimensional distributions with binary outcomes}\label{sec:two-dimensional}

Throughout the paper 
we will work concretely with binary outcome measurements in $\ZZ_{2}$.
In effect this means that our outcome space will be the {\it nerve space} of $\ZZ_2$.
This simplicial set is denoted by $N\ZZ_2$ and is defined as follows:
\begin{itemize}
\item The set of $n$-simplices is $\ZZ_2^n$,
\item The face maps are given by
$$
d_i(a_1,\cdots,a_n) = \left\lbrace
\begin{array}{ll}
(a_2,\cdots,a_n) & i=0\\
(a_1,\cdots,a_i+a_{i+1},\cdots, a_n) & 0<i<n \\
(a_1,\cdots,a_{n-1}) & i=n
\end{array}
\right.
$$
and the degeneracy maps are given by
$$
s_j(a_1,\cdots,a_n) = (a_1,\cdots,a_{{j}},0,a_{{j+1}},\cdots, a_n).
$$
\end{itemize} 

Our measurement spaces will be obtained by gluing triangles. 
A simplicial set is {\it $d$-dimensional} if all its non-degenerate simplices are in dimension $n\leq d$.
In this paper we will restrict ourselves to simplicial scenarios of the form $(X,N\ZZ_2)$ where $X$ is $2$-dimensional.
We will study simplicial distributions on such scenarios.
For simplicity of notation we will write $\sDist(X)$ omitting the outcome space when it is fixed to $N\ZZ_2${, and denote the simplicial scenario only by the measurement space $X$.}

Let us look more closely to simplicial distributions on the triangle.
Consider a triangle $\Delta^2$ with the 
generating $2$-simplex  $\sigma=\sigma^{012}$. 
A simplicial distribution is given by a simplicial set map
$$
p:\Delta^2 \to D_R(N\ZZ_2)
$$
which is determined by the distribution $p_\sigma$ on $\ZZ_2^2$. 
We will write $p_\sigma^{ab}$ for the probability of obtaining the outcome $(a,b)\in \ZZ_{2}^{2}$ when we measure $\sigma$.
The three edges bounding $\sigma$ are given by the face maps as follows: $\sigma^{01} = d_{2}\sigma^{012}$, $\sigma^{02} = d_{1}\sigma^{012}$, $\sigma^{12} = d_{0}\sigma^{012}$. 
For simplicity of notation we will write $x=\sigma^{01}$, $y=\sigma^{12}$ and $z=\sigma^{02}$.
The corresponding marginal distribution  $p_{x}:\ZZ_{2}\to\mathbb{R}_{\geq 0}$ at edge $x$ 
can be identified with $(p_{x}^{0},p_{x}^{1})$. Since $p_{x}^{0}+p_{x}^{1}=1$, it suffices just to keep $p_{x}^{0}$. 
Similarly for edges $y$ and $z$.
Compatibility with face maps requires that:
\begin{eqnarray}
p_{x}^{0} &=& p_{\sigma}^{00}+p_{\sigma}^{01},\notag\\
p_{y}^{0} &=& p_{\sigma}^{00}+p_{\sigma}^{10},\notag\\
p_{z}^{0} &=& p_{\sigma}^{00}+p_{\sigma}^{11}.\notag
\end{eqnarray}
Since $p_{\sigma}^{ab}$ is also normalized, it can be expressed by three parameters. Without loss of generality we can take these three parameters to be the marginal distributions corresponding to the edges on the boundary. 
Conversely, given the marginals on the edges  we have that
\begin{eqnarray}\label{eq:p-ab}
p_{\sigma}^{ab} &=& \frac{1}{2}\left (p_{x}^{a}+p_{y}^{b}-p_{z}^{a+b+1}\right ).
\end{eqnarray}
{Therefore a simplicial distribution on the triangle is determined by its restriction to the boundary. This observation generalizes to every $2$-dimensional simplicial set. As we will observe in Proposition \ref{pro:fConvex} for such measurement spaces restriction of a simplicial distribution to the $1$-dimensional simplicial subset consisting of all the edges determines the distribution. Alternatively, we can use the expectation coordinates instead of the probability coordinates.
For an edge $\tau\in X_{1}$, let us define its expectation value by
\begin{eqnarray}
\bar{\tau}   
= p_{\tau}^{0}-p_{\tau}^{1}.\label{eq:edge_exp}
\end{eqnarray}
Using this we can rewrite $p_{\sigma}^{ab}$, which takes the form
\begin{eqnarray}
p_{\sigma}^{ab} = \frac{1}{4}\left (1+(-1)^{a}\bar{x}+(-1)^{b}\bar{y}+(-1)^{a+b}\bar{z}\right ).\label{eq:p_in_exp}
\end{eqnarray}
}

Next we describe non-contextual distributions on $\Delta^2$. Let us start with outcome assignments. An outcome assignment $s:\Delta^2\to N\ZZ_2$ is determined by a pair of bits $s_\sigma\in \ZZ_2^2$. The corresponding deterministic distribution is ${\delta^{s}}$. For simplicity of notation we will write $\delta^{ab}$ for the deterministic distribution corresponding to the outcome assignment $s_\sigma=(a,b)$.  

\Pro{\label{pro:Delta2}
Every simplicial distribution on $\Delta^2$ is non-contextual.
}
\Proof{
Given a simplicial distribution $p:\Delta^2\to D(N\ZZ_2)$ described by $\set{p_\sigma^{ab}}_{{a,b\in\zz_2}}$. {Then} the classical distribution 
$$
d = \sum_{a,b} {d}(ab) \delta^{ab}\; \text{ with } {d}(ab)=p_\sigma^{ab}
$$
satisfies $\Theta(d)=p$.
}

{In this paper we are interested in cones of $1$-dimensional simplicial sets. For instance, the $N$-cycle scenario (Definition \ref{def:N-cycle}) is of this form.}
Given a simplicial set $X$ we will construct a new simplicial set denoted by $\Cone(X)$ which represents the topological construction of adding a new vertex and joining every $n$-simplex of $X$  to this vertex to create an $(n+1)$-simplex.
The new vertex is represented by $\Delta^0$, the simplicial set representing a point. This simplicial set is defined by
\begin{itemize}
\item {$(\Delta^0)_n=\set{\sigma^{0\cdots 0}}$,}
\item {the face and the degeneracy maps are given by deleting and copying; see Example \ref{ex:triangle}.}
\end{itemize} 
{For notational convenience we will write $c_n$ for the simplex $\sigma^{0\cdots 0}$ in dimension $n$. With this notation a face map sends $c_n\mapsto c_{n-1}$ and a degeneracy map send $c_n\mapsto c_{n+1}$.}


\begin{figure}[h!]
\centering
\begin{subfigure}{.49\textwidth}
  \centering
  \includegraphics[width=.6\linewidth]{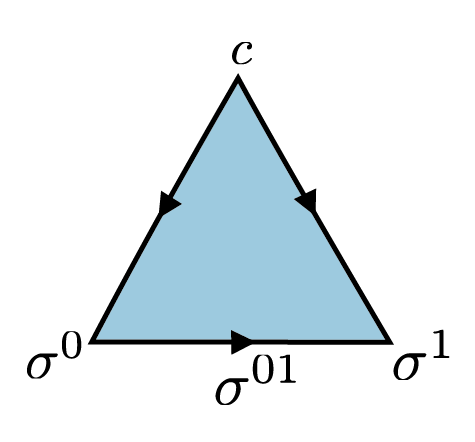}
  \caption{}
  \label{fig:cone-delta1}
\end{subfigure}%
\begin{subfigure}{.49\textwidth}
  \centering
  \includegraphics[width=.6\linewidth]{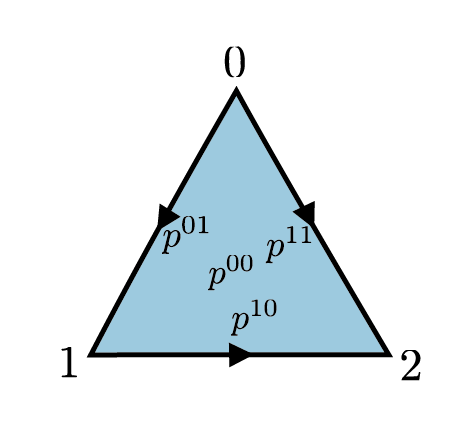}
  \caption{}
  \label{fig:dist-cone}
\end{subfigure}%
\caption{(a) A triangle can be considered as the cone of an edge. The generating $2$-simplex is given by $(c,\sigma^{01})$ whose  faces are $(c,\sigma^0)$, $(c,\sigma^1)$ and $\sigma^{01}$, where $c=c_0$. {(b) A simplicial distribution on the triangle.}
}
\label{fig:FM-topology}
\end{figure}

\Def{\label{def:cone}
The {\it cone} $\Cone(X)$ is the simplicial set given as follows:
\begin{itemize}
\item $(\Cone(X))_n = \set{c_n} \sqcup X_n \sqcup \left( \sqcup_{k+1+l=n} \set{c_k} \times X_l \right)$.
\item For $(c_k,\sigma)\in \set{c_k}\times X_l$
$$
d_i(c_k,\sigma) = \left\lbrace
\begin{array}{ll}
(c_{k-1},\sigma) & i\leq k\\
(c_k,d_{i-1-k}\sigma) & i>k
\end{array}
\right.
$$
and
$$
s_j(c_k,\sigma) = \left\lbrace
\begin{array}{ll}
(c_{k+1},\sigma) & j\leq k\\
(c_k,s_{{j}-1-k}\sigma) & {j}>k.
\end{array}
\right.
$$
Otherwise,  the face and the degeneracy maps on the $\set{c_n}$ and $X_n$ factors act the same as in $\Delta^0$ and $X$.
\end{itemize}
}

This construction is a special case of the join construction $Z\ast X$ defined for a pair of simplicial sets \cite[Chapter 17.1]{riehl2017category}. In the cone construction $Z=\Delta^0$. 
We will use the cone construction to obtain $2$-dimensional measurement spaces.

\Rem{\label{rem:nondeginCone}
{\rm
For $n \geq 1$, {the} {non-degenerate} $n$-simplices {of} $\Cone{(X)}$ {are} of the form $(c_0,\sigma)$ where $\sigma$ is a non-degenerate $(n-1)$-simplex {of} $X$. {We will usually write $c=c_0$.}
}}


\subsection{Gluing and extending distributions}
\label{sec:gluing-extending} 
 
Fundamental 
tools in the study of simplicial distributions are the extension and the gluing lemmas. They will be crucial for the proof of Fine's theorem for the $N$-cycle and the flower scenarios in Section \ref{sec:topological-proof-of-Fine}. Given a simplicial set map $f:Z\to X$ we will write 
$$
f^*:\sDist(X) \to \sDist(Z)
$$
for the map that sends a simplicial distribution $p$ on $X$ to the simplicial distribution defined by the composition $f^*p: Z\xrightarrow{f} X \xrightarrow{p} D_{{R}}(N\ZZ_2)$.  
Similarly {there is a }
map between the deterministic distributions, which is also denoted by $f^*:\dDist(X)\to \dDist(Z)$. 
In this case a deterministic distribution $\delta^s$ is sent to $\delta^{s\circ f}$.
{There is a commutative diagram}
\begin{equation}\label{dia:Theta-f*}
\begin{tikzcd}
D_R(\dDist(X)) \arrow[r,"\Theta"] \arrow[d,"D_Rf^*"] & \sDist(X) \arrow[d,"f^*"] \\
D_R(\dDist(Z)) \arrow[r,"\Theta"] & \sDist(Z)
\end{tikzcd}
\end{equation}

\Pro{\label{pro:pro4-1}
If $q\in \sDist(X)$ is non-contextual then $p=f^*(q)$ is also non-contextual.
}
\Proof{
{Let $d\in D_R(\dDist(X))$ such that $\Theta(d)=q$. Then $e=D_Rf^*(d)$ satisfies $\Theta(e)=p$ by the commutativity of Diagram (\ref{dia:Theta-f*}).
}
}

Let $A$ be a simplicial subset of $X$ and let us write $i:A\to X$ for the inclusion map. This means that each $A_n$ is a subset of $X_n$ and the simplicial structure of $A$ is compatible with that of $X$. 
Given $p\in \sDist(X)$ we will write   $p|_A$ for the distribution $i^*p$.  
For 
a deterministic distribution $\delta^s$ on $X$ the distribution $i^*\delta^s$ will be denoted by $\delta^s|_A$. Note that $\delta^s|_A = \delta^{s|_A}$ where $s|_A$ stands for the composition $s|_A: A\xrightarrow{i} X \xrightarrow{s} N\ZZ_2$.

{
An important special case of Proposition \ref{pro:pro4-1} is the following result we will  need later in the paper. 
}
\Cor{\label{cor:ExtLem} 
Let $A$ be a simplicial subset of $X$. If $q\in \sDist(X)$ is non-contextual then $q|_{A}$ is also non-contextual.
}

Another important result is the following Gluing Lemma. Using this result one can reduce the study of distributions on a measurement space to its smaller constituents in some cases.

\Lem{\label{lem:gluing} 
Suppose that $X=A\cup B$ with $A\cap B=\Delta^n$ for some $n\geq 0$. Then $p\in \sDist(X)$ is non-contextual if and only if both $p|_A\in \sDist(A)$ and $p|_B\in \sDist(B)$ are non-contextual. 
}
\Proof{
See Corollary 4.6 in \cite{okay2022simplicial}.
}



\subsection{Polytope of simplicial distributions}

Recall that the triangle $\Delta^2$ has a single generating simplex $\sigma$. The boundary $\partial \sigma$ consists of 
three non-degenerate $1$-simplices denoted by $x,y,z$. 
Using Eq.~(\ref{eq:p-ab}) the polytope of simplicial distributions $\sDist(\Delta^2)$ can be described as the space consisting of triples $(p_x^0,p_y^0,p_z^0)\in \RR^3$ satisfying
\begin{equation}\label{eq:triangle-ineq}
\begin{aligned}
p_x^0 + p_y^0 + p_z^{0}  &\geq 1  \\
p_x^0 - p_y^0 - p_z^{0}  &\geq -1  \\
-p_x^0 + p_y^0 - p_z^{0}  &\geq -1  \\
-p_x^0 - p_y^0 + p_z^{0}  &\geq -1.  \\
\end{aligned}
\end{equation}
{This set of inequalities is an example of $N$-{circle inequalities} introduced in Definition \ref{def:cycle-ineq} $(N=3)$.}
{They imply} that $\sDist(\Delta^2)$ is a tetrahedron in $\RR^3$. Proposition \ref{pro:Delta2} can be used to observe that its vertices are given by {{$(a,b,c)$} where $a,b,c\in \set{0,1}$ and} {$c+1=a+b \mod 2$}.

In general, we will show that $\sDist(X)$ is described 
as the intersection of finitely many half-space inequalities corresponding to the non-negativity of the parameters $p_{\sigma}^{ab}$.
Such a description of a polytope is called the $H$-representation.
Our goal is to characterize the geometric structure of $\sDist(X)$ including the vertices (extreme distributions) and the Bell inequalities bounding the non-contextual distributions.

\Def{\label{def:deterministic-edge}
A $1$-simplex $\tau$ of $X$ is called a {\it  deterministic edge} {(with respect to $p$)}
if  $p_\tau$ is a deterministic distribution on $\ZZ_2$.
}

\begin{pro}\label{pro:2-to-3}
If two of the edges of a triangle are deterministic then the third edge is also deterministic.
\end{pro}
\begin{proof}
Assume that $p_x^0=1$ and $p_y^0=1$,
the other cases follow similarly. 
Then the last inequality in Eq.~(\ref{eq:triangle-ineq}) implies that $p_z^0=1$.
\end{proof}

Next, we recall   
some basic facts from polytope theory \cite{chvatal1983linear,ziegler2012lectures}. 
In the $H$-representation a (convex) polytope is specified by a set of inequalities:
$$
P(A,b) = \set{x\in \RR^d:\, Ax \geq b}
$$
where $A$ is a $m\times d$ matrix and $b$ is a column vector of size $m$.
We will assume that $P\subset \RR^d$ is full-dimensional, that is, the dimension of the polytope is given by $d$. 

\Lem{\label{lem:triangle-P-dimension}
Let $X$ be a $2$-dimensional simplicial set with a single generating $2$-simplex $\sigma$ whose boundary $\partial\sigma$ consists of the $1$-simplices $x,y,z$ all of them are non-degenerate simplices. 
Consider the injective map
$$
f_\sigma:\sDist(X) \to  [0,1]^{|\partial \sigma|}
$$
that sends $p$ to the tuple $(p_\tau^0)_{\tau\in \partial \sigma}$. Then the image of $f$ is a polytope of dimension $|\partial\sigma|$.
}
\Proof{Let $P$ denote the image of $f$.
First consider the case where $|\partial \sigma|=3$. $P$ is defined by the set of inequalities in Eq.~(\ref{eq:triangle-ineq}). This is a tetrahedron  in $\RR^3$ with vertices {$(\delta^a,\delta^b,\delta^c)$ where $a,b,c \in \set{0,1}$ and} $c=a+b \mod 2$. Therefore the dimension of $P$ is $3$. Next consider $|\partial \sigma|=2$. We can assume that $x$ and $y$ identified. Then the polytope is obtained by intersecting the tetrahedron by the hyperplane $p_x^0 = p_y^0$. This gives a $2$-dimensional polytope. Finally, if $|\partial \sigma|=1$ then all the edges are identified. The polytope is obtained by intersecting the previous one with $p_y^0=p_z^0$ producing a polytope of dimension $1$. 
}

A polytope $P(A,b) \subset \RR^d$ is called full-dimensional if the dimension of the polytope is $d$.
For a simplicial set $X$ we will write $X_n^\circ$ for the set of non-degenerate simplices. 
{Let $X^{(n)}$ denote the simplicial subset of $X$ generated by $X_n^\circ$. 
For example, $X^{(1)}$ is generated by non-degenerate $1$-simplices together with the face relations coming from $X$.
} 
%
%

\Pro{\label{pro:fConvex}
Let $X$ be a  simplicial set generated by the {$2$-simplices $\sigma_1,\cdots,\sigma_k$ such that each $\partial\sigma_i$ does not contain non-degenerate edges.} The map
\begin{equation}\label{eq:map-f}
f:\sDist(X) \to \sDist(X^{(1)})=[0,1]^{|X_1^\circ|}
\end{equation}
that sends $p$ to the tuple $(p_\tau^0)_{\tau\in X_1^\circ}$ is a convex injective map. 
{Moreover, $P_X\subset \RR^{|X_1^\circ|}$ is a full-dimensional convex polytope.}
}
\Proof{
This follows from Lemma \ref{lem:triangle-P-dimension}:
For each $\sigma_i$ the restriction of $f$ to the simplicial set $X_i$ generated by $\sigma_i$ gives a map
$$
f_\sigma: \sDist(X_i) \to [0,1]^{|{\partial \sigma_i}|}.
$$ 
Thus $P_{X_i}$ is a full-dimensional polytope. 
{Consider the projection map $\RR^{|X_1^\circ|} \to {\RR^{|\partial\sigma_i|}}$ onto the coordinates of the boundary. Combining these projections we can obtain a linear embedding  $i:\RR^{|X_1^\circ|}\to \prod_{i=1}^k \RR^{|\partial\sigma_i|}$. Then $P_X$ is given by the intersection of the image of $i$ and the product of the polytopes $\prod_{i=1}^k P_{X_i}$.  This intersection remains to be full dimensional in the linear subspace. }
}

{In practice, this result implies that a simplicial distribution on a $2$-dimensional simplicial set is determined by its restriction to the edges. This description of $p$ will be referred to as the edge coordinates.}
With this result at hand, it is straight-forward to give the $H$-description of $P_X$.

\Cor{\label{cor:H-description}
Let $d_X= X_1^\circ$ and $m_X=|X_2^\circ \times \ZZ_2^2|$. We define {an} $m_X\times d_X$ matrix:
$$
A_{(\sigma;ab),\tau} = \left\lbrace
\begin{array}{ll}
(-1)^a & \tau =x \\
(-1)^b & \tau =y \\
(-1)^{a+b+1} & \tau =z \\
0 & \text{otherwise,}  
\end{array}
\right. 
$$
and a column vector $b$ of size $m_X$:
$$
b_{(\sigma,ab)} =  
\frac{\left(  (-1)^a + (-1)^b - (-1)^{a+b} \right)-1}{2}.  
$$
Then $P_X$ is described as $P(A,b)$.
}


We adopt a notation where if $\zZ\subseteq \{1,\cdots,m\}$ then $A[\zZ]$ is the matrix obtained by keeping only those rows indexed by $\zZ$ and discarding the rest, and similarly for $b[\zZ]$. Let $i\in\{1,\cdots,m\}$ index a single inequality and $x\in P$, then we call an inequality $i$ at $x$ \textit{tight} if the inequality is satisfied with equality, i.e., $A_{i}x = b_{i}$. 
For a point $x\in P$ we write $\zZ_x$ for the set of tight inequalities at $x$.

\Def{\label{def:rank} 
The {\it rank} $\rank(p)$ of a simplicial distribution $p\in \sDist(X)$ is defined to be the rank of the matrix $A[\zZ_p]$.
}

\Cor{\label{cor:vertex}
A simplicial distribution $p\in \sDist(X)$ is a vertex if and only if $\rank(p) = |X_1^\circ|$.
}
\Proof{For a full-dimensional polytope $P\subset \RR^d$, a  point $v\in P$ is a vertex if and only if it is the unique solution to $d$ tight inequalities. 
More explicitly, if $\zZ \subset \{1,\cdots,m\}$ indexes $d$ inequalities such that $A[\zZ]$ has full rank, then a vertex is given by 
$$v = A[\zZ]^{-1}b.$$ 
This basic fact applied to $P_X$, where $d = |X_{1}^{\circ}|$, combined with Proposition \ref{pro:fConvex} gives the result.
}

\subsection{Monoid structure on simplicial distributions}\label{sec:prod}


An additional algebraic feature that comes for free in
the theory of simplicial distributions is the monoid structure on $\sDist(X,Y)$ when $Y$ is a simplicial set which also has the structure of a group. 
Such a group-like simplicial set is called a simplicial group. 

Our outcome space $N\ZZ_2$ has this additional algebraic feature, which comes from the following simplicial set map:
$$
\cdot:N\ZZ_2 \times N\ZZ_2 \to N\ZZ_2
$$
defined by  
\begin{equation}\label{eq:product-nerve}
(a_1,\cdots,a_n)\cdot (b_1,\cdots,b_n) = (a_1+b_1,\cdots,a_n+b_n).
\end{equation}
{It is straight-forward to verify that this assignment respects the face and the degeneracy maps. T}his product gives the set $\dDist(X)$ of deterministic distributions the structure of a group. Given two such distributions $\delta^s$ and $\delta^r$ their product is given by $\delta^{s\cdot r}$ where 
$$
s\cdot r: X \xrightarrow{(s,r)} N\ZZ_2 \times N\ZZ_2 \xrightarrow{\cdot} N\ZZ_2.
$$ 
We will write $\delta^s\cdot \delta^r$ to denote this product of deterministic distributions.

\Lem{\label{lem:prod-delta} 
\begin{enumerate}
\item The product on $\dDist(\Delta^1)$ is given by
$$
\delta^{a} \cdot \delta^{b} = \delta^{a+b}.
$$ 
\item The product on $\dDist(\Delta^2)$ is given by
$$
\delta^{ab} \cdot \delta^{cd} = \delta^{(a+c)(b+d)}.
$$ 
\end{enumerate}
}
\Proof{ 
Let $\tau=\sigma^{01}$ denote the generating simplex of $\Delta^1$.
{Consider} two deterministic distributions $\delta^s$ and $\delta^r$ such that $s_{\tau}=a$ and $r_{\tau}=b$.  
The product $s\cdot r$  is determined by its value at $\tau$. Using Eq.~(\ref{eq:product-nerve}) we have
$$
(s\cdot r)_\tau = a\cdot b = a+b.
$$
For $\Delta^2$, we will consider the generating simplex $\sigma=\sigma^{012}$. By a similar argument applied to $s_\sigma=(a,b)$ and $r=(c,d)$ we observe that
$$
(s\cdot r)_\sigma = (a,b)\cdot (c,d) = (a+c,b+d). 
$$
}

Lemma \ref{lem:prod-delta} can be used to describe the product on $\dDist(X)$ when $X$ is $2$-dimensional. 
This product can be extended to $D(\dDist)$. Given $d,e \in D(\dDist(X))$ we define
$$
(d\cdot e)(s) = \sum_{r\cdot t=s} d(r) d(t)
$$
where the summation runs over $(\delta^r,\delta^t)\in (\dDist(X))^2$ satisfying $r\cdot t=s$.
With this product $D(\dDist(X))$ is a monoid. 
Next, we turn to the monoid structure on $\sDist(X)$. Given two simplicial distributions $p,q$ on $X$ the product $p\cdot q$ is defined by
\begin{equation}\label{eq:monoid-product}
(p\cdot q)_\sigma^a = \sum_{b+c=a} p_\sigma^b q_\sigma^c
\end{equation}
where the summation runs over 
$(b,c)\in (\ZZ_2^n)^2$ satisfying $b+c=a$.
This formula works for an $n$-simplex $\sigma$. For us the main interest is the cases $n=1,2$.


\Lem{\label{lem:prod-p-q-low-dim}
Let $X$ be a simplicial set and $p,q\in \sDist(X)$. 
\begin{enumerate}
\item For $\tau \in X_1$, we have 
$$
(p\cdot q)^{0}_\tau=p^{0}_\tau \cdot  q^{0}_\tau+ p^{1}_\tau \cdot  q^{1}_\tau, \;\;\;\; (p\cdot q)^{1}_\tau=p^{0}_\tau \cdot  q^{1}_\tau+ p^{1}_\tau \cdot  q^{0}_\tau
$$ 
\item For $\sigma \in X_2$, we have 
$$
(p\cdot q)^{00}_\sigma=p^{00}_\sigma \cdot  q^{00}_\sigma+ p^{01}_\sigma \cdot  q^{01}_\sigma
+p^{10}_\sigma \cdot  q^{10}_\sigma+p^{11}_\sigma \cdot  q^{11}_\sigma, \;\; \;\; 
(p\cdot q)^{01}_\sigma=p^{00}_\sigma \cdot  q^{01}_\sigma+ p^{01}_\sigma \cdot  q^{00}_\sigma
+p^{10}_\sigma \cdot  q^{11}_\sigma+p^{11}_\sigma \cdot  q^{10}_\sigma
$$ 
$$
(p\cdot q)^{10}_\sigma=p^{00}_\sigma \cdot  q^{10}_\sigma+ p^{01}_\sigma \cdot  q^{11}_\sigma
+p^{10}_\sigma \cdot  q^{00}_\sigma+p^{11}_\sigma \cdot  q^{01}_\sigma, \;\; \;\; 
(p\cdot q)^{11}_\sigma=p^{00}_\sigma \cdot  q^{11}_\sigma+ p^{01}_\sigma \cdot  q^{10}_\sigma
+p^{10}_\sigma \cdot  q^{01}_\sigma+p^{11}_\sigma \cdot  q^{00}_\sigma
$$ 
\end{enumerate}
}
\Proof{
Follows directly from Eq.~(\ref{eq:monoid-product}).
}

Moreover, the map $\Theta:D(\dDist(X)) \to \sDist(X)$ is a homomorphism of monoids. For more on the monoid structure and its interaction with convexity see \cite{kharoof2022simplicial}.  
We will use the action of the group $\dDist(X)$ on the monoid $\sDist(X)$ that comes from the product in Eq.~(\ref{eq:monoid-product}).
Explicitly,  for $\sigma \in X_2$ and $\tau \in X_1$ this action is described as follows:
\begin{equation}\label{eq:action}
(\delta^{a}\cdot q)_\tau^{c}=q_{\tau}^{c+a},\;\;\;\; 
(\delta^{ab}\cdot q)_\sigma^{cd}=q_{\sigma}^{(c+a)(d+b)}.
\end{equation}
{Note that this action maps vertices of $\sDist(X)$ to vertices.}

 
\begin{pro}\label{pro:action}
\begin{enumerate}
\item For two non-contextual simplicial distributions $p$ and $q$ in $\sDist(X)$, the product $p \cdot q$ is a non-contextual  distribution.

\item  A simplicial distribution 
$p \in \sDist(X)$ is non-contextual if and only if $\delta^{s} \cdot p$ is non-contextual.

\item  A simplicial distribution 
$p \in \sDist(X)$ is vertex if and only if $\delta^{\varphi} \cdot p$ is a vertex. 
\end{enumerate}
\end{pro}
\begin{proof}
Part $1$ follows from the fact that the map $\Theta:D(\dDist(X)) \to \sDist(X)$ is a homomorphism of monoids \cite[Lemma 5.1]{kharoof2022simplicial}.  
Part $2$ follows from part $1$. 
\end{proof}

Part (2) of this proposition implies that the action of $\dDist(X)$ on $\sDist(X)$ maps a (non)-contextual vertex to a (non)-contextual vertex. 
We describe the action in the case of the well-known CHSH scenario in {Example \ref{ex:chsh} below}.

The following simplicial distributions on {$\Delta^2=\Cone(\Delta^1)$} will play a distinguished {in later sections when we study $2$-dimensional scenarios more closely:}
\begin{equation}\label{eq:p-plus-minus}
p_+^{ab} = \left\lbrace
\begin{array}{ll}
1/2 & b=0\\
0   & \text{otherwise.}
\end{array}
\right.
\;\;\;\;
p_-^{ab} = \left\lbrace
\begin{array}{ll}
0 & b=0\\
1/2   & \text{otherwise.}
\end{array}
\right.
\end{equation}
{We follow the convention in Fig.~(\ref{fig:dist-cone}).}

\Def{\label{def:G-plus-minus} 
Let $X$ be a $1$-dimensional simplicial set. 
We will write $G_\pm(CX)$ for the subset of  simplicial distributions $p\in \sDist(CX)$ satisfying $p|_{(c,\tau)}=p_{\pm}$ for every $\tau\in X_1^\circ$.
}

Next we show that this set is a group.
We will denote the distribution in $G_\pm(CX)$ with $p|_{(c,\tau)}=p_{+}$ for every $\tau\in X_1^\circ$ by 
$e_{+}$.

\Pro{\label{pro:G-plus-minus}
$G_\pm(CX)$ is an abelian group with $e_{+}$ as the identity. In addition, every element has order $2$, that is,
$$
G_\pm(CX) \cong \ZZ_2^{X_{1}^\circ}.
$$ 
}
\Proof{ 
By part $2$ of Lemma \ref{lem:prod-p-q-low-dim} we have
$$
p_{+}\cdot p_{+}=p_{+} \;\; , \;\; p_{+}\cdot p_{-}=p_{-} \;\; , \;\; p_{-}\cdot p_{-}=p_{+} \;\; , \;\;
$$
Therefore the statement holds for $X=\Delta^1$, that is we have
$$
G_\pm(\Cone(\Delta^1)) \cong \ZZ_2.
$$
Now for arbitrary $X$ and $p,q\in G_\pm(CX)$ the product is computed triangle-wise, i.e., $(p\cdot q)_\sigma = p_\sigma\cdot q_\sigma$. Therefore the statement easily generalizes.
}

\Ex{\label{ex:chsh}
{\rm 
The CHSH scenario  consists of four triangles organized into a disk with vertices $v_0,v_1,w_0,w_1$ and $c$. For each pair {$(v_i,w_j)$}  there is an edge which we denote by {$\tau_{ij}$}. This constitutes the boundary of the disk. 
{There are four non-degenerate triangles $\sigma_{ij}=(c,\tau_{ij})$ as depicted in Fig.~(\ref{fig:cone-chsh}). The interior edges  $(c,v_i)$ and $(c,w_j)$ will be denoted by $x_i$ and $y_j$, respectively.}
This scenario is a particular case of the $N$-cycle scenario in Definition \ref{def:N-cycle}.
Here $N$ is the number of edges on the boundary, hence in this case $N=4$.
{Using the edge coordinates of Proposition \ref{pro:fConvex} a} simplicial distribution $p$ on the CHSH scenario can be 
described by the tuple  
$(p_{x_0},p_{\tau_{00}}
,p_{y_0},p_{\tau_{10}},
p_{x_1},p_{\tau_{11}},
p_{y_1},p_{\tau_{01}})$. 
It is well-known that $p$ is non-contextual if and only if it satisfies the CHSH inequalities \cite{chsh69}:
\begin{equation}\label{eq:CHSH-ineq-edge}
\begin{aligned}
0\leq  p_{\tau_{00}}^0+p_{\tau_{10}}^0+p_{\tau_{11}}^0-p_{\tau_{01}}^0 \leq 2 \\
0\leq  p_{\tau_{00}}^0+p_{\tau_{10}}^0-p_{\tau_{11}}^0+p_{\tau_{01}}^0 \leq 2 \\
0\leq  p_{\tau_{00}}^0-p_{\tau_{10}}^0+p_{\tau_{11}}^0+p_{\tau_{01}}^0 \leq 2 \\
0\leq - p_{\tau_{00}}^0+p_{\tau_{10}}^0+p_{\tau_{11}}^0+p_{\tau_{01}}^0 \leq 2  
\end{aligned}
\end{equation}
Also the contextual vertices are known. They are given by the Popescu--Rohrlich (PR) boxes \cite{pr94}:
A PR box is a simplicial distribution $p$ such that 
$p_{\sigma_{ij}}=p_\pm$ 
for {$i,j \in \set{0,1}$} with the further restriction that the number of $p_-$'s is odd.

\begin{figure}[h!]
\centering
\includegraphics[width=.4\linewidth]{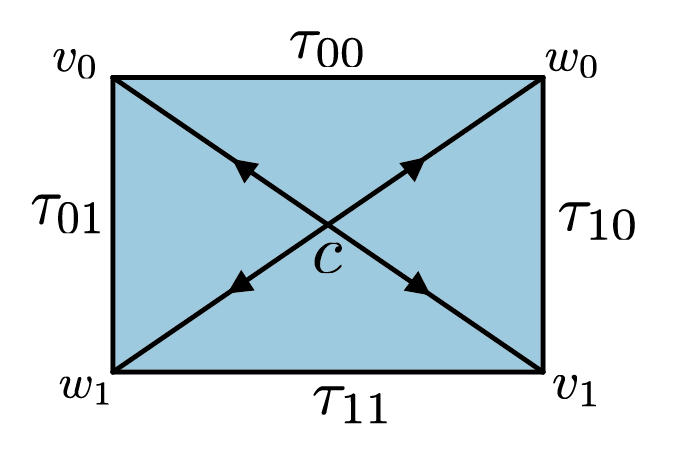}
\caption{The CHSH scenario as the cone of the circle  consisting of $\tau_{00},\tau_{01},\tau_{10},\tau_{11}$ (Definition \ref{def:circle}).  
}
\label{fig:cone-chsh}
\end{figure}

We begin with the action on the PR boxes. 
{By Eq.~(\ref{eq:action}) we see that $\delta^{00}$ and $\delta^{10}$ are the only deterministic distributions on the triangles that fix $p_{+}$ and $p_{-}$. From this observation we conclude that among the $16$ deterministic distributions on the CHSH scenario the ones that fix a given PR box are  
$(\delta_{\sigma_{00}}^{00},\delta_{\sigma_{01}}^{00},\delta_{\sigma_{10}}^{00},\delta_{\sigma_{11}}^{00})$ and $(\delta_{\sigma_{00}}^{10},\delta_{\sigma_{01}}^{10},\delta_{\sigma_{10}}^{10},\delta_{\sigma_{11}}^{10})$.}
Thus the size of the orbit is $16/2=8$, which gives all the PR boxes.

To describe the action of $\delta^s$ on the Bell inequalities we need to switch {back}
to the edge coordinates.
For notational convenience we will write $p_i$ for the $i$-th entry of this tuple.
Then the deterministic distribution $\delta^s$ is given by $(\delta^{a_0},\delta^{a_0+b_0},\delta^{b_0},\delta^{a_1+b_0},\delta^{a_1},\delta^{a_1+b_1},\delta^{b_1},\delta^{a_0+b_1})$.
Using the notational convenience introduced above, in these coordinates the action of $\delta^s$ on $p$ is given by  
$$ 
(\delta^s\cdot p)_i=  (\delta^s)_i \cdot p_i.
$$
Now, substituting these new values to the Bell inequality gives the action. For example, 
 the action of  $(\delta^{0},\delta^{1},\delta^{1},\delta^{0},\delta^{1},\delta^{1},\delta^{0},\delta^{0})$  on the Bell inequality 
\begin{equation}\label{eq:sample-bell}
p_{\sigma_{00}}^0 
+p_{\sigma_{10}}^0 +p_{\sigma_{11}}^0 -p_{\sigma_{01}}^0  \leq 2
\end{equation} 
gives 
$
1-p_{\sigma_{00}}^0  
+p_{\sigma_{10}}^0 +1-p_{\sigma_{11}}^0 -p_{\sigma_{01}}^0  \leq 2
$, which can be put in a more familiar form 
$$
p_{\sigma_{00}}^0  
-p_{\sigma_{10}}^0 +p_{\sigma_{11}}^0 +p_{\sigma_{01}}^0  \geq 0.
$$  
We can compute the stabilizer of the Bell inequality in Eq.~(\ref{eq:sample-bell}). The relevant edge coordinates are $\sigma_{ij}$ that constitute the boundary of the CHSH scenario. 
The relevant coordinates of $\delta^s$ that can change the inequality are $\delta^{a_i+b_j}$ where $i,j\in \set{0,1}$. 
Then the stabilizer consists of those deterministic distributions that satisfy {$a_i+b_j=0\mod 2$ for every $i,j\in \set{0,1}$.} The size of this group is $2$ and therefore there are $8=16/2$ elements in the orbit. This covers all $8$ of the Bell inequalities.
}
}

See Section \ref{sec:application-vell-ineq} for more on the action on Bell inequalities.

\section{Distributions on the classical $N$-disk}


The classical $N$-disk scenario has the measurement space given by a disk triangulated in a way that results in only non-contextual (or classical) distribution.

\Def{\label{def:classical-N-disk}
For $N\geq 3$ let $D_N$ denote the following simplicial set:
\begin{itemize}
\item Generating $2$-simplices: $\sigma_1,\cdots,\sigma_{N-2}$.
\item Identifying relations:
$$
d_{j_1}(\sigma_1)=d_{j_2}(\sigma_2),\; d_{j'_2}(\sigma_2)=d_{j_3}(\sigma_3),\; d_{j'_3}(\sigma_3)=d_{j_4}(\sigma_4)\; \cdots\;  d_{j'_{N-3}}(\sigma_{N-3})=d_{j_{N-2}}(\sigma_{N-2})
$$
where $j_1,j_2,j'_2,\dots,j_{N-3},j'_{N-3}
,j_{N-2} \in \{0,1,2\}$ and 
$j_k \neq j'_k$ for $2 \leq k \leq N-3$.
\end{itemize}   
}
%
%
The classical $N$-disk can be constructed by successive gluing. To see this, starting from an initial non-degenerate simplex $\sigma_{1}$ we successively glue simplices along a single edge so that $\sigma_{i+1}$ shares a single common edge with $\sigma_{i}$, terminating with the simplex $\sigma_{N-2}$. 
{The simplices $\sigma_{1}$ and $\sigma_{N-2}$ in any classical $N$-disk will be referred to as the initial and terminal simplices, respectively.}
In particular, the gluing described by the face relations is such that the boundary of the disk has $N$ edges and form an $N$-{circle} in the sense of Definition \ref{def:circle}. Letting $(\partial D_{N})_{1}^{\circ}$ be the non-degenerate edges on the boundary of the classical $N$-disk, non-degenerate simplicies in the classical $N$-disk are distinguished by
\begin{eqnarray}
|(\partial D_{N})_{1}^{\circ}\cap (\sigma_{i})_{1}^{\circ}| = \left\lbrace
\begin{array}{ll}
2 & i = 1, N-2 \\
1   & \text{otherwise.}
\end{array}
\right.
\;\;\;\;
\label{eq:initial-terminal-triangles}
\end{eqnarray}
{Such edges in the classical $N$-disk are called boundary edges, otherwise we call them interior edges. The classical $3$-disk is $\Delta^2$, while the diamond space $D$ is an example of a classical $4$-disk.}
See Fig.~(\ref{fig:classical6Disk}) for an example of a classical $6$-disk.
    
\begin{figure}[h!]
\centering
\includegraphics[width=.6\linewidth]{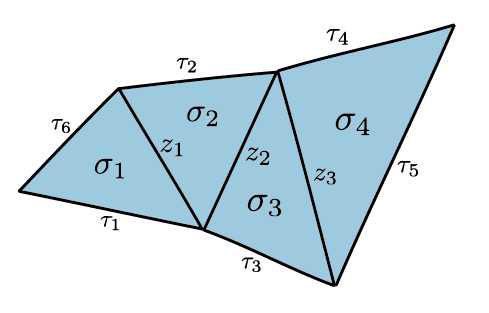}
\caption{Classical $6$-disk {with initial simplex $\sigma_1$ and terminal simplex $\sigma_4$.}
}
\label{fig:classical6Disk}
\end{figure}

\Pro{\label{pro:Ndisk}
Any simplicial distribution on the classical $N$-disk scenario is non-contextual.
}
\Proof{
This follows from (Gluing) Lemma \ref{lem:gluing} since $D_N$ is constructed by gluing $N$ triangles along a $\Delta^1$. At each step we can apply the Gluing Lemma.
}

\subsection{Fourier--Motzkin elimination}
\label{sec:FM} 
 
As is well-known, systems of linear equations can be solved using Gaussian elimination. For systems of linear inequalities there exists a related technique known as Fourier-Motzkin (FM) elimination; see e.g., \cite{ziegler2012lectures}. A linear inequality in $d$ variables can be written as $a^{T}x \geq b$, where $a, x \in \mathbb{R}^{d}$ and $b\in \mathbb{R}$. For $m$ such linear inequalities we have $a_{i}^{T}x\geq b_{i}$ ($i = 1,\cdots,m$). Taking each vector $a_{i}^{T}$ to be a row of a matrix $A$, this set of $m$ inequalities can be compactly written as $Ax\geq B$, where $A \in \mathbb{R}^{m\times d}$ and $B\in \mathbb{R}^{m}$. The feasible region defined by $Ax\geq B$ (if one exists) forms a polyhedron.

To perform FM elimination of a variable $x_{j}$, let us first index all inequalities where $x_{j}$ appears with positive, negative, or zero coefficients as $\iI_{j}^{+}$, $\iI_{j}^{-}$, and $\iI_{j}^{0}$, respectively. We then solve for $x_{j}$:
\begin{eqnarray}
x_{j} &\geq & \frac{B_{i}}{a_{ij}} - \sum_{k\neq j}\frac{a_{ik}}{a_{ij}}x_{k},\quad \forall i\in \iI_{j}^{+},\notag\\
x_{j} &\leq & -\frac{B_{i}}{|a_{ij}|} + \sum_{k\neq j}\frac{a_{ik}}{|a_{ij}|}x_{k},\quad \forall i\in \iI_{j}^{-}.\notag
\end{eqnarray}
Then for every $(i, i^{\prime}) \in \iI_{j}^{+}\times \iI_{j}^{-}$ we have that such an $x_{j}$ exists so long as
\begin{eqnarray}
\frac{B_{i}}{a_{ij}} - \sum_{k\neq j}\frac{a_{ik}}{a_{ij}}x_{k} \leq x_{j} \leq -\frac{B_{i^{\prime}}}{|a_{i^{\prime}j}|} + \sum_{k\neq j}\frac{a_{i^{\prime}k}}{|a_{i^{\prime}j}|}x_{k},\notag
\end{eqnarray}
which is equivalent to
\begin{eqnarray}
\frac{B_{i}}{a_{ij}} - \sum_{k\neq j}\frac{a_{ik}}{a_{ij}}x_{k} \leq -\frac{B_{i^{\prime}}}{|a_{i^{\prime}j}|} + \sum_{k\neq j}\frac{a_{i^{\prime}k}}{|a_{i^{\prime}j}|}x_{k}.
\end{eqnarray}
This can be rearranged to give a new set of inequalities in $d-1$ variables whose solution, should it exist, is the same as the original set of inequalities.

\subsubsection{Application to the Diamond scenario}


As a warm up we begin by considering the diamond scenario $D$ described in Example \ref{ex:diamond}. 
We will adapt a more convenient notation for the generating simplices of the two triangles $A$ and $B$. The first one will be denoted by 
$\sigma^{012}$ and the other one by $\sigma^{01'2}$.  
The diamond $D$ is obtained by gluing $A$ and $B$ along the $d_1$ face, i.e., the simplex $\sigma^{02}$. 
Again for ease of notation the probabilities $p_{\sigma_A^{012}}^{ab}$ and $p_{\sigma_B^{01'2}}^{a'b'}$ will be denoted by $p_{012}^{ab}$ and $p_{01'2}^{a'b'}$; respectively.
In this section we will use the expectation coordinates introduced in Eq.~(\ref{eq:edge_exp}) and (\ref{eq:p_in_exp}). 
These eight probabilities 
that 
are
required to be non-negative
is equivalent (up to an overall constant factor) to the inequalities
\begin{eqnarray}
&&1+(-1)^{a}\bar{\sigma}^{01}+(-1)^{b}\bar{\sigma}^{12}+(-1)^{a+b}\bar{\sigma}^{02}\geq 0,\label{eq:diamond_ineqa}\\
&&1+(-1)^{a'}\bar{\sigma}^{01'}+(-1)^{b'}\bar{\sigma}^{1'2}+(-1)^{a'+b'}\bar{\sigma}^{02}\geq 0,
\label{eq:diamond_ineqb}
\end{eqnarray}
for all $a,b,a',b'\in \ZZ_{2}$.

\begin{pro}\label{pro:chsh}
Let $D$ be a diamond and $\partial D$ denote its boundary. Then a distribution $p\in \sDist(\partial D)$ extends to $\tilde{p} \in \sDist(D)$ if and only if the CHSH inequalities are satisfied:
\begin{equation}\label{eq:CHSH}
(-1)^{a}\bar{\sigma}^{01}+(-1)^{b}\bar{\sigma}^{12}+(-1)^{a'}\bar{\sigma}^{01'}+(-1)^{b'}\bar{\sigma}^{1'2} \geq 0
\end{equation}
where $a,a'b,b'\in \ZZ_2$ satisfying $a+b+a'+b'=1~\text{mod}~2$.
\end{pro}
\Proof{
Proof of this result is given in {\cite[Proposition 4.10]{okay2022simplicial}.}
We provide an exposition here for completeness.
All of the coefficients that appear in Eqns.~(\ref{eq:diamond_ineqa}-\ref{eq:diamond_ineqb}) are just $\pm 1$, to perform FM elimination it suffices to sum up the inequalities where $\sigma^{02}$ has positive and negative coefficient. For inequalities coming from the same triangle this just yields that $-1\leq \bar{\sigma}^{ij}\leq 1$ --- we call such inequalities \emph{trivial}. When we combine inequalities from different triangles we obtain 
the inequalities in Eq.~(\ref{eq:CHSH}).
}
 
\begin{rem}
{\rm
We can interpret FM elimination geometrically as deleting an edge from a topological space; see Fig.~(\ref{fig:FM-topology})
}
\end{rem}


\begin{figure}[h!]
\centering
\begin{subfigure}{.49\textwidth}
  \centering
  \includegraphics[width=.3\linewidth]{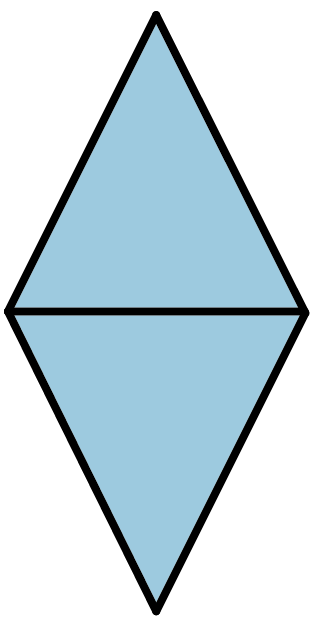}
  \caption{}
  \label{fig:FM-topology-a}
\end{subfigure}%
\begin{subfigure}{.49\textwidth}
  \centering
  \includegraphics[width=.3\linewidth]{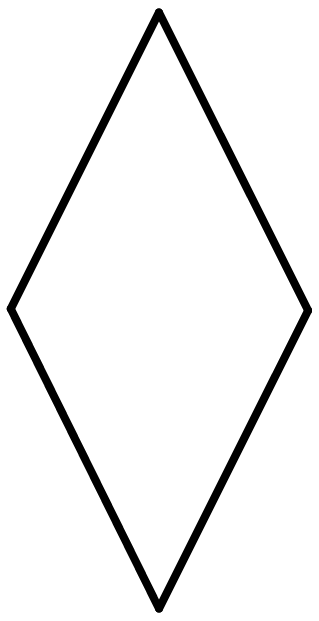}
  \caption{}
  \label{fig:FM-topology-b}
\end{subfigure}%
\caption{{(a) The classical $4$-disk is the diamond scenario. (b) FM elimination is interpreted geometrically as deleting an edge from a topological space.}
}
\label{fig:FM-topology}
\end{figure}

In an abuse of terminology, we will sometimes say that we eliminate an edge $\sigma$, when what we actually   mean is that we perform FM elimination on the corresponding expectation value $\bar{\sigma}$ that appears in the inequalities.


\subsection{Extending to the classical $N$-disk}

\begin{defn} \label{def:cycle-ineq}
Let $\tau_1,\cdots,\tau_N$ denote the generating edges on the boundary of $D_N$. 
We define the
$N$-{circle inequalities} by
\begin{eqnarray}
0 \leq n-2 + \sum_{i = 1}^{n} (-1)^{a_{i}}\bar{\tau}_{i}
\label{eq:n-cycle-ineq-FM}
\end{eqnarray}
where $\sum_{i=1}^{n}a_{i} = n+1~\text{mod}~2$. 
\end{defn}
\noindent 

\Ex{\label{ex:3cycle-exp}
{\rm
Clearly a 
triangle
$\Delta^2$ is just a classical $3$-disk and the $3$-{circle inequalities} come from Eq.~(\ref{eq:p_in_exp}): 
$$
p_x^a + p_y^b -p_z^{a+b+1} \geq 0.
$$
Note also that the diamond space is an example of a classical $4$-disk and the CHSH inequalities correspond to the $4$-{circle inequalities}.
}
}



\begin{lem}%
\label{lem:same-set}
Consider a set of $N$-{circle inequalities} with a common coordinate $\bar z$. 
Applying FM elimination to $\bar z$ the 
resulting inequalities are satisfied if the remaining coordinates $\bar{\tau}_{i}$ each satisfies $-1\leq \bar{\tau}_{i}\leq 1$.
\end{lem}
\begin{proof}
Consider two inequalities where $\bar{z}$ appears with opposite signs
\begin{eqnarray}
&& 0 \leq n-2 + \bar{z} + \sum_{i = 1}^{n-1} (-1)^{a_{i}}\bar{\tau}_{i} \notag\\
&& 0 \leq n-2 -\bar{z} + \sum_{i = 1}^{n-1} (-1)^{b_{i}}\bar{\tau}_{i}  \notag
\end{eqnarray}
where 
$\sum_{i=1}^{n-1}a_{i} = n+1~\textrm{mod}~2$ and $\sum_{i=1}^{n-1}b_{i} = n~\textrm{mod}~2$.
To perform FM elimination of $\bar{z}$ we add these inequalities together and observe that due to the conditions on $a_{i}$ and $b_{j}$ that at least one other variable will cancel after summing. Let $N\subset \{1,\cdots,n-1\}$ index all variables that do \emph{not} cancel. (Note that $|N|\leq n-2$.) The inequality after summing becomes
\begin{eqnarray} \label{eq:after-summing}
0 \leq 2\left (n-2 +\sum_{i \in N}(-1)^{a_{i}}\bar{\tau}_{i} \right ),
\end{eqnarray}
or equivalently
\begin{eqnarray}
0 \leq \left (n-2 -|N|\right ) +\sum_{i \in N}\left (1+(-1)^{a_{i}}\bar{\tau}_{i} \right ).\notag
\end{eqnarray}
Using Eq.~(\ref{eq:edge_exp}) we see that these inequalities are 
satisfied 
if $-1\leq \bar{\tau}_{i}\leq 1$ for all $i\in N$:
This condition gives us
\begin{eqnarray}
0\leq \frac{n-2-|N|}{2}+\sum_{i\in N}p_{\tau_{i}}^{a_{i}},\notag
\end{eqnarray}
where each term is non-negative since $|N|\leq n-2$ and $0\leq p_{\tau_{i}}^{a_{i}}\leq 1$. Thus the inequalities in Eq.~(\ref{eq:after-summing}) are satisfied.
\end{proof}

\begin{lem}%
\label{lem:combining-ineq}
Suppose we have a set of $N$-{circle} and $M$-{circle inequalities} that overlap on only a single variable $\bar{z}$. FM elimination of $\bar{z}$ yields a set of $(N+M-2)$-{circle inequalities} (plus trivial inequalities).
\end{lem}

\begin{proof}
We begin by noting that if we sum up inequalities coming from the same set of {circle inequalities} then by Lemma~\ref{lem:same-set} we get trivial inequalities. Let us consider the other case where $\bar{z}$ comes from two different sets; see Fig.~(\ref{fig:FM-cycle}). First note that there are $2^{K-1}$ ($K\geq 1$) inequalities in a set of $K$-{circle inequalities}. Let $\iI_{M}^{\pm}$ index the $M$-{circle inequalities} where $\bar{z}$ has a positive (or negative) coefficient and observe that $|\iI_{M}^{\pm}| = 2^{M-2}$. Similarly for $\iI_{N}^{\pm}$. FM elimination proceeds by summing up inequalities indexed by $(i,i')\in \iI_{M}^{+}\times \iI_{N}^{-}$ and $(j,j') \in \iI_{M}^{-}\times \iI_{N}^{+}$. This amounts to $2\times 2^{(M-2)+(N-2)} = 2^{(N+M-2)-1}$ new inequalities, which is precisely the amount needed for a set of $(N+M-2)$-{circle inequalities}.

\begin{figure}[h!]
\centering
\begin{subfigure}{.49\textwidth}
  \centering
  \includegraphics[width=.8\linewidth]{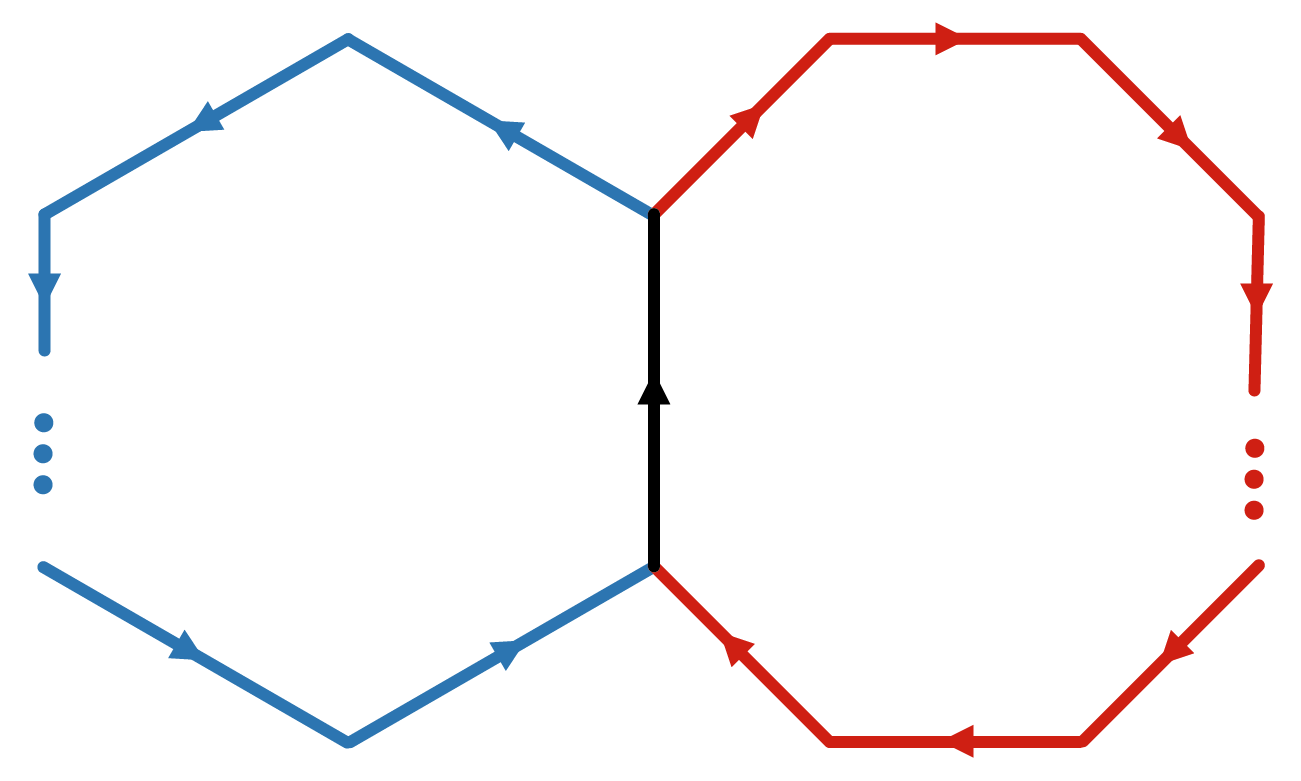}
  \caption{}
  \label{fig}
\end{subfigure}%
\begin{subfigure}{.49\textwidth}
  \centering
  \includegraphics[width=.8\linewidth]{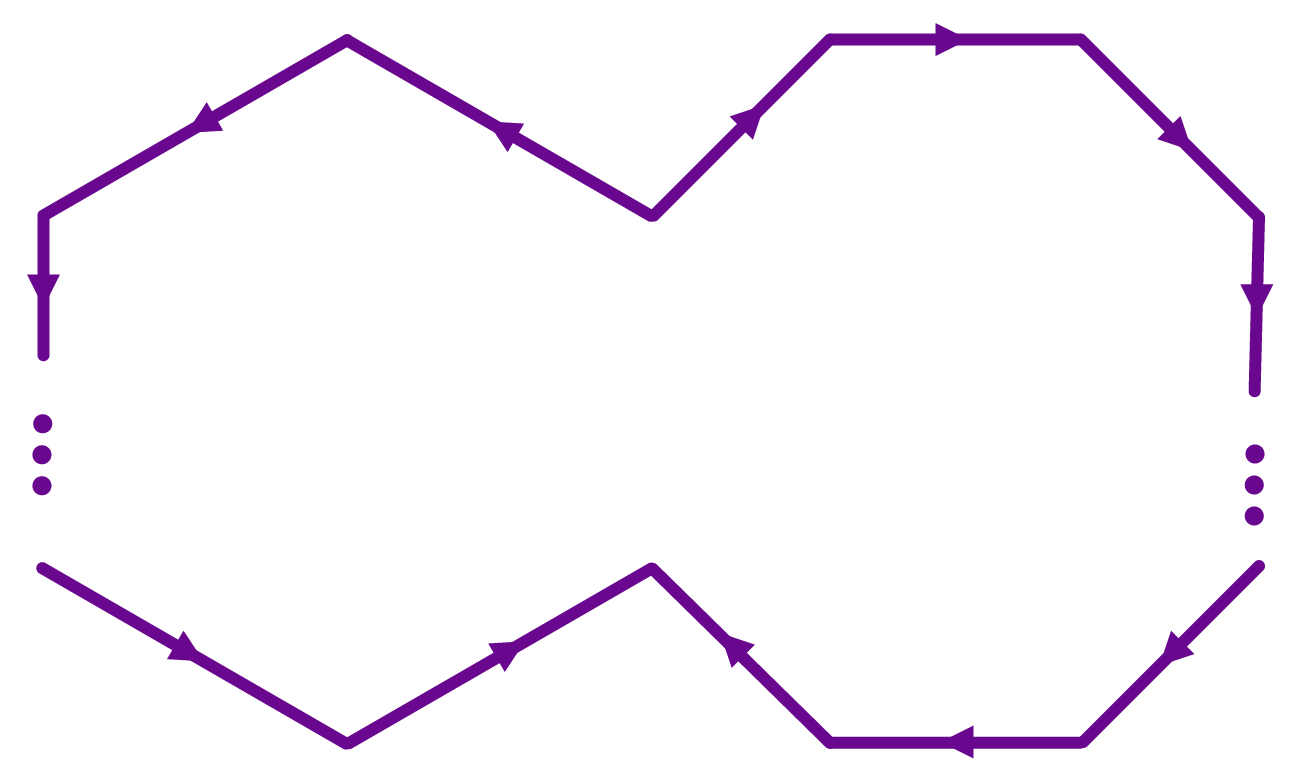}
  \caption{}
  \label{fig}
\end{subfigure}%
\caption{ 
FM elimination of an edge {(black)} common to an $N$-{circle} and $M$-{circle} in (a)  yields an 
{inequality for the $N+M-2$-{circle} in (b).}
}
\label{fig:FM-cycle}
\end{figure}


To find the precise form of these inequalities, let us consider explicitly two inequalities indexed by $(i,i') \in \iI_{M}^{-}\times \iI_{N}^{+}$. We denote the variables appearing in the $M$-{circle} and $N$-{circle inequalities} as $\tau_{j}$ ($j = 1,\cdots,M$) and $\tau'_{k}$ ($k=1,\cdots,N$), respectively, and denote $\bar{z} = \tau_{M} = \tau'_{N}$. Summing the two inequalities we obtain
\begin{eqnarray}
0&\leq& N-2 + \bar{z} + \sum_{k = 1}^{N-1}(-1)^{a_{k}}\tau'_{k}+\left ( M-2 - \bar{z} + \sum_{j = 1}^{M-1}(-1)^{b_{j}}\tau_{j}\right),\notag
\end{eqnarray}
where $\sum_{k=1}^{N-1}a_{k} = N+1~\text{mod}~2$ and $\sum_{j= 1}^{M-1}b_{j} = M~\text{mod}~2$. This is equivalent to
\begin{eqnarray}
0&\leq & (N+M-2) -2 +  \sum_{k = 1}^{N-1}(-1)^{a_{k}}\tau'_{k} + \sum_{j = 1}^{M-1}(-1)^{b_{j}}\tau_{j}
\end{eqnarray}
where $\sum_{k=1}^{n-1}a_{k} + \sum_{j=1}^{M-1}b_{j} = N+M+1~\text{mod}~2$.
Noting that $N+M+1~\text{mod}~2 = (N+M-2)+1~\text{mod}~2$, this is precisely an $(N+M-2)$-{circle inequality}. A similar argument holds for $\iI_{M}^{+}\times \iI_{N}^{-}$ and this proves the result.
\end{proof}

We have the following corollary of Lemma \ref{lem:combining-ineq}:

\begin{cor}%
\label{cor:n-and-3-FM}
Suppose we have a set of $N$-{circle} and $3$-{circle inequalities} that overlap on only a single variable $\bar{z}$. FM elimination of $\bar{z}$ yields a set of $(N+1)$-{circle inequalities} (plus trivial inequalities). {See Fig.~(\ref{fig:FM-cycle}). 
}
\end{cor}

Next we apply these preliminary results to the classical $N$-disk scenario. 

\begin{pro}%
\label{pro:n-cycle-ineq-FM}%
A distribution $p\in \sDist(\partial D_N)$ extends to a distribution $\tilde p$ on $D_N$ if and only if the $N$-{circle inequalities} (and the trivial inequalities $-1\leq \bar{\tau}_{i}\leq 1$) are satisfied.
\end{pro}
\Proof{
We consider a classical $N$-disk (e.g., see Fig.~(\ref{fig:classical6Disk})) 
such that the edges on the boundary are labeled by $\tau_{i}$ ($i=1,\cdots,n$) and those on the interior are denoted $z_{j}$ ($j =1,\cdots,n-3$). 
For the first part of our proof, our strategy is to perform FM elimination successively on the interior edges $z_{j}$ beginning\footnote{It is well known that the order in which FM elimination is performed does not affect the final result, however, a ``bad" ordering can lead to an explosion in intermediate inequalities to keep track of; see e.g., \cite{fukuda2005double}.} with $z_{1}$ and ending in $z_{n-3}$. Consider the two classical $3$-disks bounded by $\{\tau_{1},\tau_{n}, z_{1}\}$ and $\{z_{1},\tau_{2}, z_{2}\}$, respectively. 
By Corollary~\ref{cor:n-and-3-FM}, FM elimination of $\bar{z}_{1}$ yields a set of $4$-{circle inequalities} (plus trivial inequalities) together with the remaining inequalities in which $\bar{z}_{1}$ does not appear. For each successive application of FM for the edges $z_{j}$ we can apply Corollary~\ref{cor:n-and-3-FM}. 
After $n-3$ iterations we are left with an $n$-{circle inequality}, as well as trivial inequalities. This proves one direction. On the other hand, FM elimination guarantees that we can find a set of $\{z_{i}:\,i=1,\cdots,n-3\}$ such that we can reverse this process and extend from the boundary to the $n$-order disk.
}

\subsection{Bouquet of classical $N$-disks}

It is possible to extend Proposition~\ref{pro:n-cycle-ineq-FM} slightly by considering the union of $N$ disks of varying 
size.

{
\Def{\label{def:circle}
Let $C_1$ denote the simplicial set with a single generating $1$-simplex $\tau$ with the relation
$$
d_0\tau=d_1\tau.
$$ 
For $N\geq 2$, let $C_N$ denote the $1$-dimensional simplicial set consisting of the generating $1$-simplices $\tau_1,\cdots,\tau_N$ together with the identifying relations
$$
d_{i'_1}\tau_1 = d_{i_2}\tau_2,\; d_{i_2'}\tau_2= d_{i_3}\tau_3\;\cdots\; d_{i'_N} \tau_N = d_{i_1} \tau_1 
$$
where $i_k\neq i_k' \in \set{0,1}$ for $1\leq k\leq N$.
We call $C_N$ the {\it $N$-circle} space, or simply the circle space when $N=1$; see Fig.~(\ref{fig:8-circle}).  
A 
circle
 of length $N$ on a simplicial set $X$ is given by an injective simplicial set map $C_N\to X$. We will also write $C_N$ for the image of this map.
} 
}

\begin{figure}[h!]
\centering
\begin{subfigure}{.33\textwidth}
  \centering
  \includegraphics[width=\linewidth]{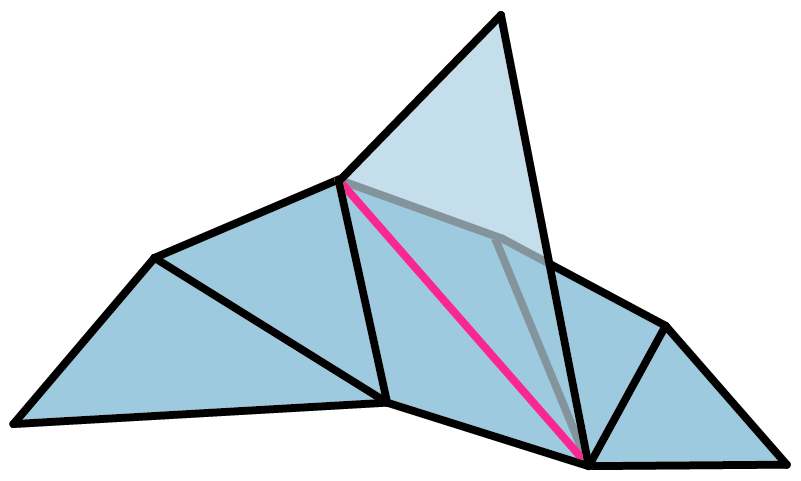}
  \caption{}
  \label{fig:bouquet-a}
\end{subfigure}%
\begin{subfigure}{.33\textwidth}
  \centering
  \includegraphics[width=\linewidth]{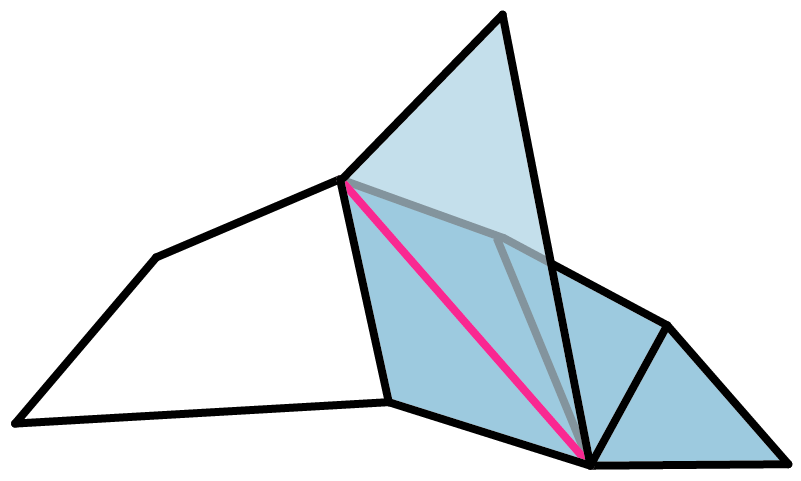}
  \caption{}
  \label{fig:bouquet-b}
\end{subfigure}%
\begin{subfigure}{.33\textwidth}
  \centering
  \includegraphics[width=\linewidth]{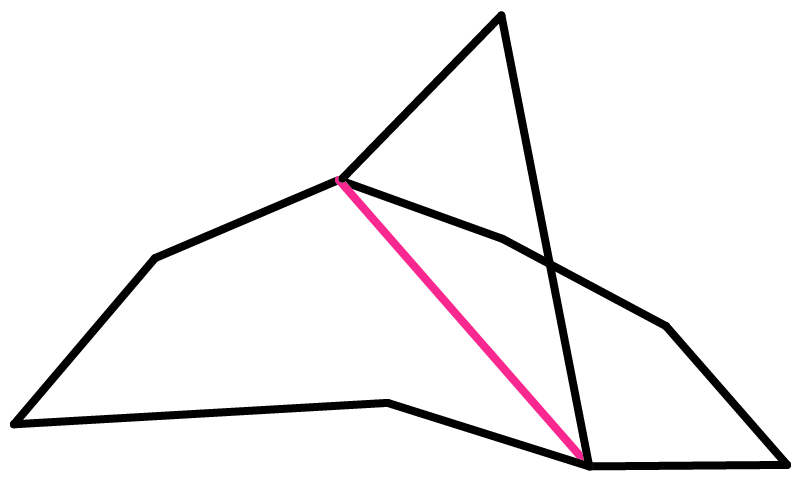}
  \caption{}
  \label{fig:bouquet-c}
\end{subfigure}%
\caption{{(a) A bouquet of classical $N$-disks is constructed by gluing each $D_{N_{i}}$ along a common edge (purple) that is a boundary edge of its initial triangle. (b) For each disk we perform FM elimination beginning with the interior edge of the terminal triangle. (c) We continue until only the common edge (purple) remains.}
}
\label{fig:bouquet}
\end{figure}

\begin{cor}%
\label{cor:generalized-extension}
Let $X$ be a $2$-dimensional simplicial set 
obtained by gluing $D_{N_1},\cdots,D_{N_k}$ along a common edge,
and let $\partial X$ be the $1$-dimensional
given by the boundary of $X$.
Then $p\in \sDist(\partial X)$ 
extends to a distribution $\tilde p$ on $X$  
if and only if for every 
{circle} $C_{M_j} \subset \partial X$, where $j=1,\cdots,{k\choose 2}$, 
we have that $p|_{C_{M_j}}$ 
satisfies the corresponding 
$M_j$-{circle inequality}. 
\end{cor}
\begin{proof}
For each $D_{N_{i}}$ its initial and terminal triangles, which we denote by $\sigma_{1}^{(i)}$ and $\sigma_{N_{i}-2}^{(i)}$, respectively, are distinguished via Eq.~(\ref{eq:initial-terminal-triangles}). A bouquet $X$ of classical $N$-disks is then constructed by gluing each $D_{N_{i}}$ along the single common edge $\tau$, which we take to be either boundary edge of the initial triangle $\sigma_{1}^{(i)}$; see Fig.~(\ref{fig:bouquet-a}).

For each disk $D_{N_{i}}$ in $X$ we perform FM elimination on all interior edges, beginning with the interior edge of the terminal triangle and concluding with the interior edge of the initial triangle; see Fig.~(\ref{fig:bouquet-b}). The ordering of which disks FM elimination is applied to is arbitrary and does not affect the calculations. For each disk we stop before eliminating the edge $\tau$; see Fig.~(\ref{fig:bouquet-c}). By Proposition~\ref{pro:n-cycle-ineq-FM} this will result in $N$-{circle inequalities} (plus trivial inequalities), each corresponding to a {circle} of length $N_{i}$. Since $\tau$ appears in all $k$ sets of {circle inequalities}, and it is the only edge in the intersection of these {circle}s, then Lemma~\ref{lem:combining-ineq} applies. We will have ${k \choose 2}$ {circle}s $C_{M_{j}}$, where $j = 1,\cdots,{k\choose 2}$.
\end{proof}

The extension result of Corollary~\ref{cor:generalized-extension} will be useful in proving Fine's theorem in Section \ref{sec:topological-proof-of-Fine} for various types of scenarios. 

\begin{figure}[h!]
\centering
\begin{subfigure}{.49\textwidth}
  \centering
  \includegraphics[width=.5\linewidth]{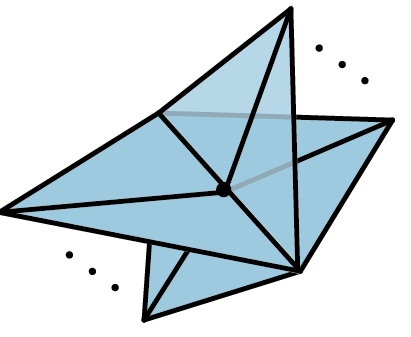}
  \caption{}
  \label{fig:2m22-bell}
\end{subfigure}%
\begin{subfigure}{.49\textwidth}
  \centering    
  \includegraphics[width=.5\linewidth]{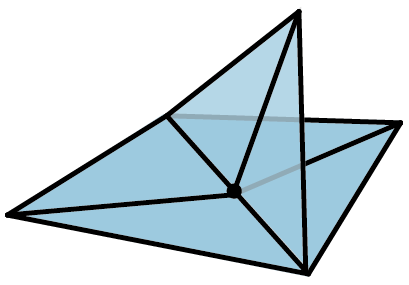}
  \caption{}
  \label{fig:2322-bell}
\end{subfigure}\\
\begin{subfigure}{.49\textwidth}
  \centering
  \includegraphics[width=.5\linewidth]{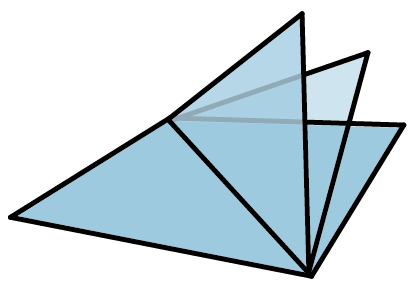}
  \caption{}
  \label{fig:2322-classical}
\end{subfigure}%
\begin{subfigure}{.49\textwidth}
  \centering
  \includegraphics[width=.5\linewidth]{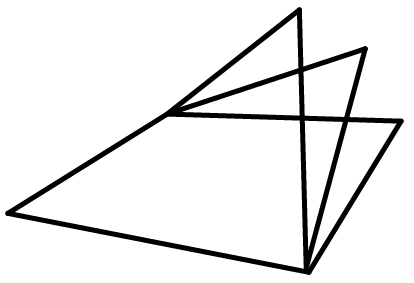}
  \caption{}
  \label{fig:2322-class-FM}
\end{subfigure}%
\caption{(a) Topological realization of the $(2,m,2,2)$ Bell scenario. (b) Measurement space for the $(2,3,2,2)$ Bell scenario. (c) The bouquet of classical $N$-disks (here, all $3$-disks) used in proof of Fine's theorem for the $(2,3,2,2)$ Bell scenario. (d) A distribution on the $(2,3,2,2)$ Bell scenario is classical if and only if $6 = {4 \choose 2}$ sets of CHSH inequalities are satisfied, corresponding to the $6$ possible {circle}s.
}
\label{fig:bouquet}
\end{figure}

\begin{ex}\label{ex:special-case-for-Fine}
{ \rm
Bipartite {$(m_A,m_B,d_A,d_B)$} Bell scenarios  consist of parties Alice and Bob performing one of $m_{A}$, $m_{B}$ measurements with one of $d_{A}$, $d_{B}$ outcomes, respectively. 
In \cite{collins2004relevant} it was shown that, by generalizing an argument due to Fine \cite{fine1982hidden,fine1982joint}, that the CHSH inequalities are also necessary and sufficient for this more general scenario.

{A topological realization for $(2,m,2,2)$ 
 Bell scenario is given in Fig.~(\ref{fig:2m22-bell}). Note that this scenario is a special case of the flower scenario depicted in Fig.~(\ref{fig:flower}). 
In Theorem \ref{thm:bouquet-of-lines} we will generalize Fine's characterization of non-contextual distributions to to flower scenarios. 
The basic idea of our approach can be sketched in the case of $(2,3,2,2)$ Bell scenario; see Fig.~(\ref{fig:2322-bell}). In this case we use the bouquet of $3$-disks depicted in  Fig.~(\ref{fig:2322-classical}). By Corollary \ref{cor:generalized-extension} a distribution on the boundary Fig.~(\ref{fig:2322-class-FM}) extends to the whole space} if and only if the $6 = {4 \choose 2}$ sets of $4$-{circle inequalities} are satisfied.

 


}
\end{ex}

\section{Distributions on the $N$-cycle scenario and beyond}




The measurement space of the $N$-cycle scenario is a disk triangulated into $N$ triangles as in Fig.~(\ref{fig:8-cycle}).

\Def{\label{def:N-cycle}
Let $\tilde C_N$ denote the following simplicial set:
\begin{itemize}
\item Generating $2$-simplices: $\sigma_1,\cdots,\sigma_N$.
\item Identifying relations:
$$
d_{i'_1}\sigma_1 = d_{i_2}\sigma_2,\; d_{i_2'}\sigma_2= d_{i_3}\sigma_3\;\cdots\; d_{i'_N} \sigma_N = d_{i_1} \sigma_1 
$$
where $i_k\neq i_k' \in \set{1,2}$ for $1\leq k\leq N$.


\end{itemize}   
$\tilde C_1$ has a single generating simplex $\sigma$ with the identifying relation $d_1\sigma = d_2\sigma$. 
}

Topologically $\tilde C_N$ is obtained from its boundary, which is circle consisting of $N$ edges, by introducing a new point, the vertex in the middle, and coning off the boundary.
This construction will be very useful in our analysis of simplicial distributions on $2$-dimensional measurement spaces.

\begin{figure}[h!]
\centering
\begin{subfigure}{.49\textwidth}
  \centering
  \includegraphics[width=.5\linewidth]{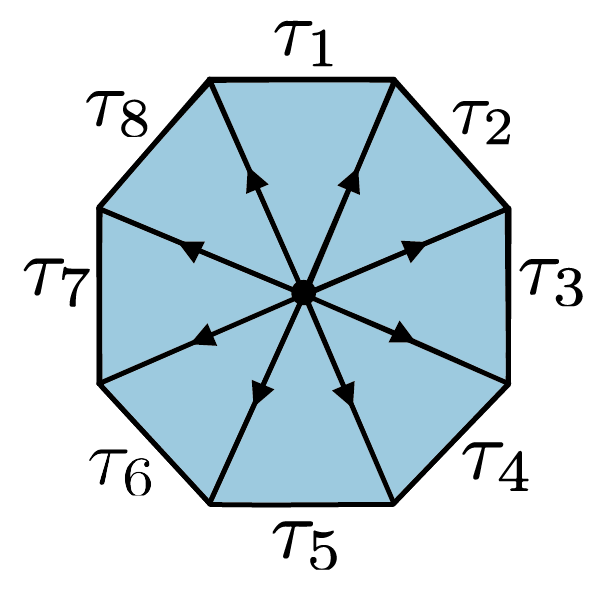}
  \caption{}
  \label{fig:8-cycle}
\end{subfigure}%
\begin{subfigure}{.49\textwidth}
  \centering
  \includegraphics[width=.5\linewidth]{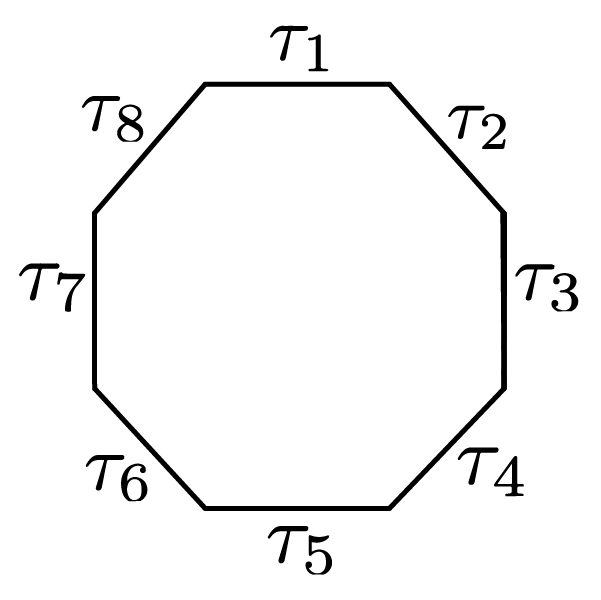}
  \caption{}
  \label{fig:8-circle}
\end{subfigure}
\caption{
The $N$-cycle scenario ($N=8$) depicted in (a) is the cone of the $N$-circle scenario depicted in (b) consisting of the edges $\tau_1,\cdots ,\tau_8$ on the boundary.  
}
\label{fig:8-circle-cycle}
\end{figure}



Observe that  $\tilde C_N$ is precisely the cone of $C_N$. 
In later sections, we will study the vertices of the polytope of simplicial distributions on this scenario and describe the Bell inequalities bounding the non-contextual distributions. 
Note that $\tilde C_1$ is a new scenario in the sense that it cannot be realized in the conventional picture of non-signaling distributions, such as in the language of sheaf theory \cite{abramsky2011sheaf}. This is the smallest space on which a contextual simplicial distribution is defined.


\subsection{Topological proof of Fine's theorem
}
\label{sec:topological-proof-of-Fine}

The proof of Fine's theorem for the CHSH scenario given in   \cite[Theorem 4.13]{okay2022simplicial} 
relies on topological methods. 
Here we show that these methods can be generalized to other interesting scenarios including the $N$-cycle scenario and the flower scenario {obtained by gluing cycle scenarios as in Fig.~(\ref{fig:flower}).}


\begin{lem}\label{lem:DetCone}
Let $X$ be a simplicial set. The map 
$$\dDist(\Cone(X)) \to {\zz_2^{|X_0|}}$$ that send $\delta^s$ to {$(s(c,v))_{v\in X_0}$} is {a bijection}.
\end{lem} 
\proof{
{
{A deterministic distribution on $\Delta^2$} 
is given by an assignment $(x,y,z)\mapsto (a,b,c)$ such that $a+b+c=0\mod 2$. Therefore in $\Cone(X)$ once the edges $(c,v)$ are assigned an outcome the remaining edges will be determined.
}
}
\begin{lem}\label{lem:OneDiric}
Let $X$ be a simplicial set.
Given a non-contextual distribution $p\in \sDist(\Cone(X))$, the restriction $p|_{C_N}$ to an $N$-{circle} $C_N\subset X$ satisfies the $N$-{circle inequalities}.
\end{lem}
\begin{proof} 
Let $D_N$ be a classical $N$-disc with $C_N$ as the boundary. Recall that we can think of $\tilde C_N$ as the cone of $C_N$.
Note that $(C_N)_0=(D_N)_0$, thus using Lemma \ref{lem:DetCone} we obtain that the map $\dDist(\Cone(D_N))) \to \dDist(\tilde{C}_N)$  induced by the inclusion $C_N {\to} D_N$ is an isomorphism. 
We have the commutative diagram 
\begin{equation}\label{SQQQ}
\begin{tikzcd}
D_{R}(\dDist(\Cone(D_N))  
\arrow[rr,"\Theta"]
 \arrow[d,"\cong"'] && \sDist(\Cone (D_N))   
 \arrow[d,""] \\
D_{R}(\dDist(\tilde{C}_N)) 
 \arrow[rr,"\Theta"] && 
 \sDist(\tilde{C}_N)
\end{tikzcd}
\end{equation}
The simplicial distribution $p$ is non-contextual, thus by {{Corollary \ref{cor:ExtLem}}}
$p|_{\tilde{C}_N}$ is also non-contextual. 
Therefore by Diagram (\ref{SQQQ}) the distribution $p|_{\tilde{C}_N}$ can be extended to a distribution on $\Cone(D_N)$.
In particular, $p|_{C_N}$ extended to a distribution on $D_N$. By 
Proposition \ref{pro:n-cycle-ineq-FM}
we obtain the result.
\end{proof}
%
%

\begin{figure}
\centering
\begin{subfigure}{.33\textwidth}
  \centering
  \includegraphics[width=.6\linewidth]{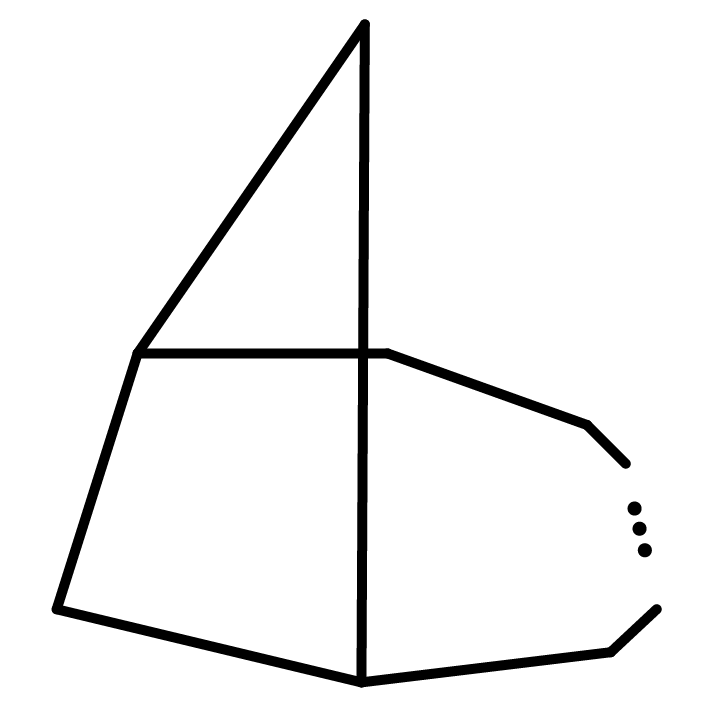}
  \caption{}
  \label{fig:ZPic}
\end{subfigure}%
\begin{subfigure}{.33\textwidth}
  \centering
  \includegraphics[width=.6\linewidth]{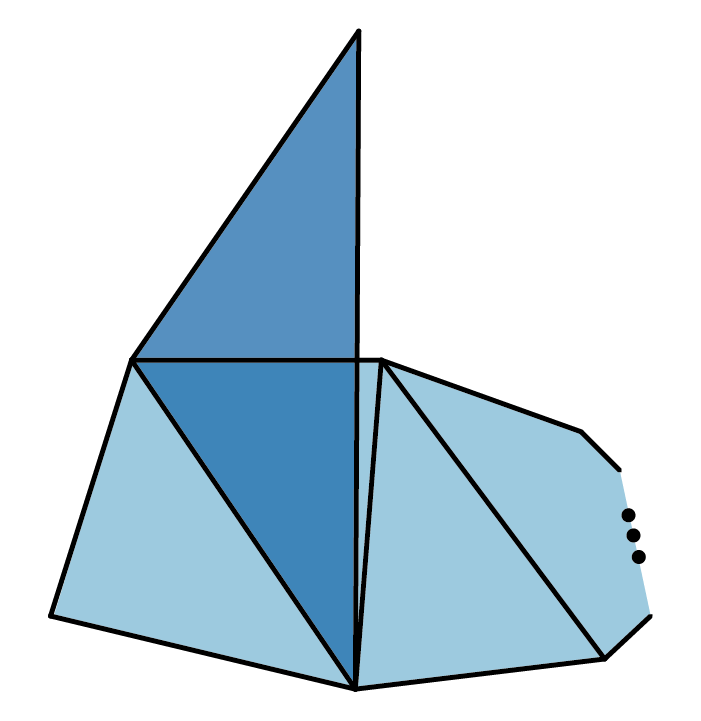}
  \caption{}
  \label{fig:WPic}
\end{subfigure}
\begin{subfigure}{.33\textwidth}
  \centering
  \includegraphics[width=.6\linewidth]{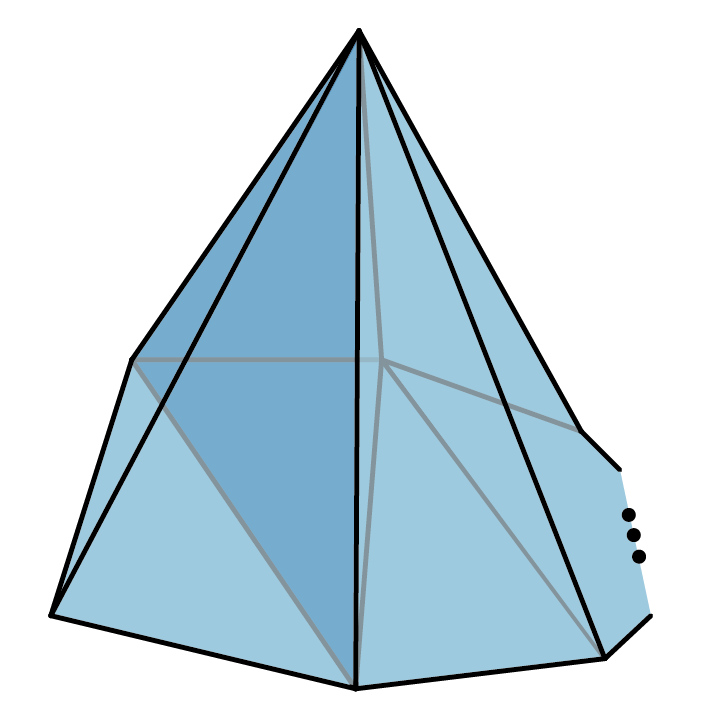}
  \caption{}
  \label{fig:cycle-V}
\end{subfigure}
\caption{{
(a) $Z$ is obtained from the $N$-circle by attaching two edges. (b) $W$ is obtained from a classical $N$-disk by attaching a triangle. (c) $V$ is obtained from $W$ by attaching an $N$-cycle space.
}
}
\label{fig:subspaces-CX}
\end{figure}

\begin{pro}\label{pro: CiBellll}
A distribution $p\in \sDist(\tilde C_N)$ is non-contextual if and only if $p|_{C_N}$ satisfies the $N$-{circle inequalities}.  
\end{pro}
\begin{proof}
Forward direction is proved in Lemma \ref{lem:OneDiric}. For the converse, we will need the following simplicial {sets}:
\begin{itemize}
\item  $Z\subset \tilde{C}_N$ denotes the $1$-dimensional simplicial subset obtained by gluing two edges to $C_N$ as depicted in   Fig.~(\ref{fig:ZPic}).
\item $W$ is obtained by gluing a triangle $\Delta^2$ to $D_N$ as in  Fig.~(\ref{fig:WPic}).
\item  $V$ is the simplicial set obtained by gluing $\tilde{C}_N$ and $W$ along $Z$ as in Fig.~(\ref{fig:cycle-V}):
\begin{equation}\label{eq:V}
V = \tilde{C}_N\cup_Z W.
\end{equation}
\end{itemize}
Let $p$ be a simplicial distribution on $\tilde C_N$ such that $p|_{C_N}$ satisfies the $N$-{circle inequalities}. By Proposition \ref{pro:n-cycle-ineq-FM} we conclude that the restriction of $p$ on the other two {{circle}s} on $Z$ also satisfies the {circle inequalities}. Therefore
Corollary \ref{cor:generalized-extension} implies that $p|_{Z}$ extends to a simplicial distribution $q$ on $W$. Since $p$ and $q$ match on $Z$ there is a simplicial distribution $P$ on $V$ such that $P|_{\tilde C_N}=p$ and $P|_{W}=q$.
Now, consider the following decomposition given in Fig.~(\ref{fig:LMRPic}):
$$
V = L \cup_{M} R
$$ 
where $L=\partial \Delta^3$ is the boundary of the tetrahedron, $M=\Delta^2$ a triangle.  
Let $\tilde{R}$ denote  the {simplicial subset} of $R$ obtained by omitting the bottom triangles. {Note that $\tilde R$ is an $(N-1)$-cycle scenario.} 
Let $Q$ be the restriction of $P$ on $L \cup_M \tilde{R}$. 
The simplicial distribution $Q$ is non-contextual if and only if $Q|_L$ and $Q|_{\tilde{R}}$ are both non-contextual by (Gluing) Lemma  \ref{lem:gluing}.
By {\cite[Proposition 4.12]{okay2022simplicial}} every simplicial distribution on $\partial\Delta^3$ is non-contextual. 
Therefore it suffices to have $Q|_{\tilde{R}}$  non-contextual to guarantee that $Q$ is non-contextual. In fact, $Q|_{\tilde{R}}$ is the restriction {of} $P|_{R}$ {on $\tilde R$}, thus by Proposition \ref{pro:n-cycle-ineq-FM} we conclude that $Q$ satisfies the {circle inequalities} on the lower {circle} of $\tilde{R}$. By induction on $N$, we conclude that $Q|_{\tilde{R}}$ is non-contextual. Finally, by {{{Corollary \ref{cor:ExtLem}}}}
this implies that $p=Q|_{\tilde{C}_N}$ is non-contextual.   

\end{proof}

\begin{figure}
\centering
  \includegraphics[width=.35\linewidth]{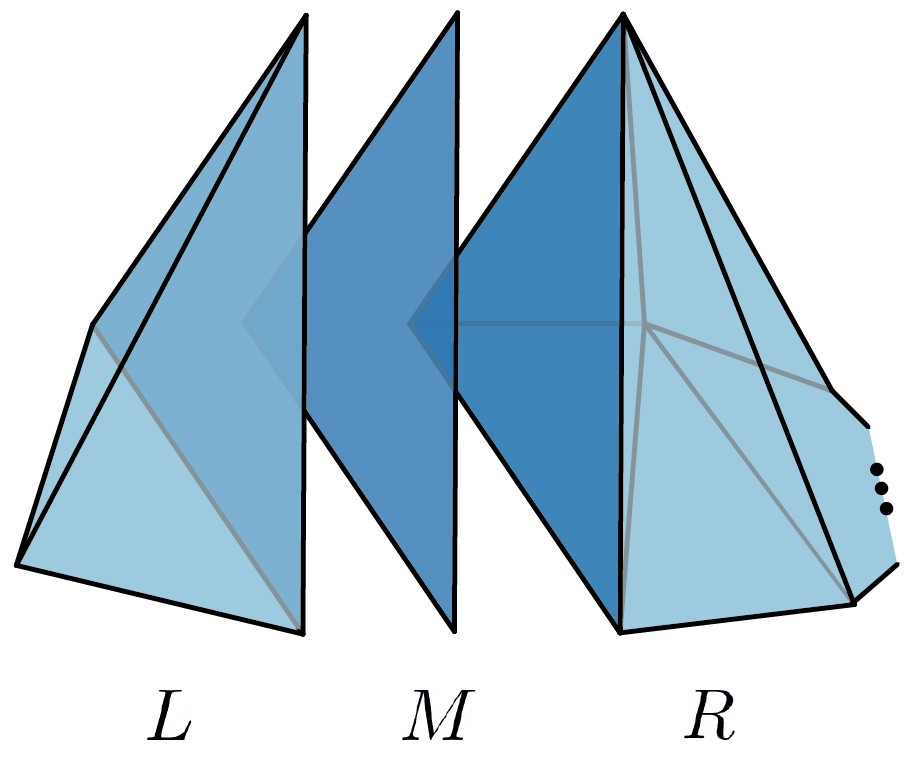}
\caption{{
The space $V$ partitioned into three {simplicial subsets} $L$, $M$ and $R$.
}
}
\label{fig:LMRPic}
\end{figure}

Combining this result with Proposition \ref{pro:n-cycle-ineq-FM} gives the following result which will be used in the generalization of Proposition \ref{pro: CiBellll} to the flower scenario.

\begin{cor}\label{col:CycleExt} 
A distribution $p$ on $\tilde C_N$ 
is non-contextual if and only if it extends to a distribution on $\tilde C_N \cup_{C_N} D_N$.  
\end{cor}

Taking $N=3$ in Corollary \ref{col:CycleExt} gives us a sufficient and necessary condition when a distribution on the boundary of a triangle can be extended to a distribution on the triangle. Thus we obtain a useful result that characterizes the image of the map 
$f:\sDist(X)\to \sDist(X^{(1)})$
introduced in (\ref{eq:map-f}). {(We remark that the following result still holds when the restriction that $\partial\sigma_i$ does not contain non-degenerate edges is removed {in Proposition \ref{pro:fConvex}}.)}

\Cor{\label{cor:image-f}
Let $X$ be a $2$-dimensional simplicial set. 
A distribution $q\in \sDist(X^{(1)})$ is in the image of the map $f$ in (\ref{eq:map-f}) if and only if $(q^0_{d_0\sigma},
q^0_{d_1\sigma},q^0_{d_2\sigma})$ satisfies the $3$-{circle inequality} for all $\sigma\in X_2$.
}

To generalize Proposition \ref{pro: CiBellll} to the flower scenario we need a stronger version of the Gluing lemma.

\begin{lem}\label{lem:GluLem}
Let $X=\cup_{i=1}^m A_i$ such that $A_i \cap A_j=\Delta^n$ for every $i \neq j$. Then $p \in
\sDist(X)$ is non-contextual if and only if $p|_{A_i}$ is non-contextual for every $1\leq i \leq m$.   
\end{lem}
\begin{proof}
Follows by induction and Lemma \ref{lem:gluing}.
\end{proof}
%

The flower scenario is obtained by gluing lines at their end points.
Let $L_N$ denote the simplicial set consisting of the generating $1$-simplices $\tau_1,\cdots,\tau_N$ together with the identifying relations
$$
d_{i'_1}\tau_1 = d_{i_2}\tau_2,\; d_{i_2'}\tau_2= d_{i_3}\tau_3\;\cdots\; d_{i'_{N-1}} \tau_{N-1} = d_{i_N} \tau_N 
$$
where $i_k\neq i_k' \in \set{0,1}$. Topologically this simplicial set represents a line of length $N$.
%

\Def{\label{def:flower}
Let {$X(N_1,\cdots,N_k)$} denote the simplicial set obtained by gluing the lines {$L_{N_1},\cdots,L_{N_k}$} at their boundary, i.e., the two terminal points $v$ and $w$; see Fig.~(\ref{fig:X}).
We will call the cone of {$X(N_1,\cdots,N_k)$} a {\it flower scenario}.
}

\begin{figure}
\centering
\begin{subfigure}{.33\textwidth}
  \centering
  \includegraphics[width=.9\linewidth]{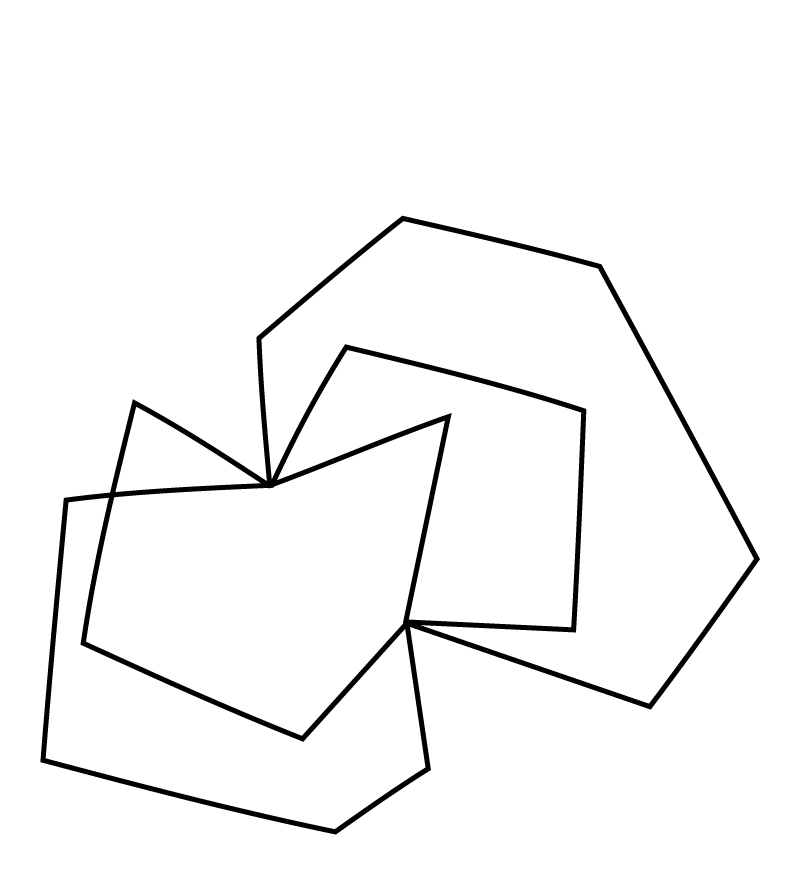}
  \caption{}
  \label{fig:X}
\end{subfigure}%
\begin{subfigure}{.33\textwidth}
  \centering
  \includegraphics[width=.9\linewidth]{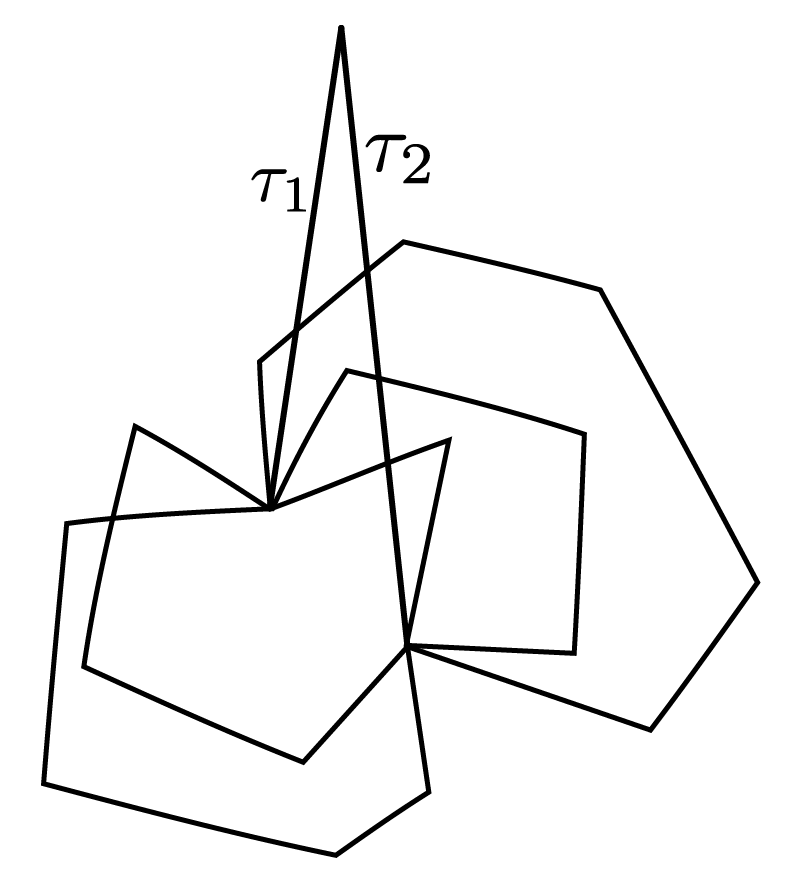}
  \caption{}
  \label{fig:Z}
\end{subfigure}
\begin{subfigure}{.33\textwidth}
  \centering
  \includegraphics[width=.9\linewidth]{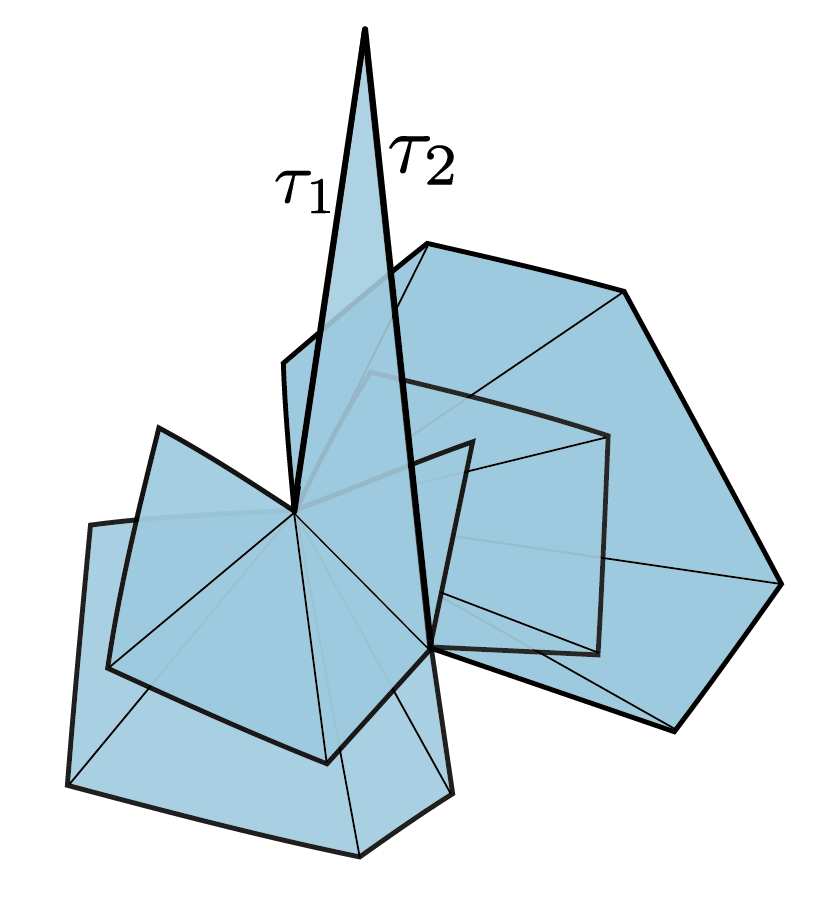}
  \caption{}
  \label{fig:W}
\end{subfigure}
\caption{(a) $X$ is obtained by gluing lines at their end points. (b)   $Z$ comes with an additional piece consisting of two edges. (c)  $W$ is obtained by filling the {circle}s in $Z$ with classical disks.
}
\label{fig:subspaces-CX-flower}
\end{figure}

\begin{thm}\label{thm:bouquet-of-lines} 
{Let $\Cone(X)$ denote the flower scenario where $X=X(N_1,\cdots,N_k)$}.  
A distribution $p\in \sDist(\Cone(X))$ is non-contextual if and only if for every {{circle}} $C_N$ on $X$ the restriction $p|_{C_N}$ satisfies the $N$-{circle inequalities}.
\end{thm}
\begin{proof}
Forward direction follows from Lemma \ref{lem:OneDiric}.
For the converse, we introduce the following simplicial sets:
\begin{itemize}
\item $Z\subset \Cone(X)$ denotes the $1$-dimensional simplicial set obtained by gluing two edges $\tau_1$ and $\tau_2$ to $X$ as depicted in Fig.~(\ref{fig:Z}).
\item $W$ is obtained by filling in the {circle}s in $Z$ by classical disks as in Fig.~(\ref{fig:W}).
\item Gluing $\Cone(X)$ with $W$ along the intersection $Z$ we obtain
$$
V = \Cone(X)\cup_Z W.
$$
\end{itemize}
Let $p$ be a simplicial distribution on $\Cone(X)$ satisfying the {circle inequalities} for every {circle} in $X$.
Moreover, by Proposition \ref{pro:n-cycle-ineq-FM} the distribution $p$ also satisfies the {circle inequalities} for the remaining {circle}s in the larger space $Z$ since on these {circle}s the distribution extends to classical disks {contained in $\Cone(X)$}.
Then Corollary \ref{cor:generalized-extension}  implies that $p|_Z$ extends to a simplicial distribution $q$ on $W$. 
The two distributions $p$ and $q$ give a distribution $P$ on $V$.
Now, we define the following simplicial subsets of $V$:
\begin{itemize}
\item $M$ 
{denotes the 
triangle in Fig.~(\ref{fig:W}) with two of the edges given by $\tau_1$ and $\tau_2$. }

\item {$\tilde V_i$ is obtained by gluing $\Cone(L_{N_i})$ and $M$ along $\tau_1$ and $\tau_2$.} 

\item {$ V_i$ is obtained by gluing  $\tilde V_i$ and the classical disk contained in $W$ whose boundary coincides with {$\partial \tilde V_i$}; see Fig.~(\ref{fig:Vi}).}  
 
\end{itemize}
Note that $M=\tilde V_i \cap \tilde V_j= V_i\cap V_j$ for distinct $i,j$. In addition, we obtain another decomposition of $V$ as given in Fig.~(\ref{fig:Vi}):
$$
V =  V_1 \cup_M  V_1 \cup_M \cdots \cup_M  V_k.
$$
By Corollary \ref{col:CycleExt} the simplicial distribution $P|_{\tilde V_i}$ is non-contextual since it is the restriction of  $P|_{ V_i}$.
Therefore by Lemma \ref{lem:GluLem} the distribution $P|_{\tilde  V_1 \cup_M  \tilde V_2 \cup_M \cdots \cup_M \tilde V_k.}$ is non-contextual, so by {{{Corollary \ref{cor:ExtLem}}}} 
the restriction $p=P|_{\Cone(X)}$ is non-contextual.
\end{proof}

\begin{figure} 
  \centering
  \includegraphics[width=.8\linewidth]{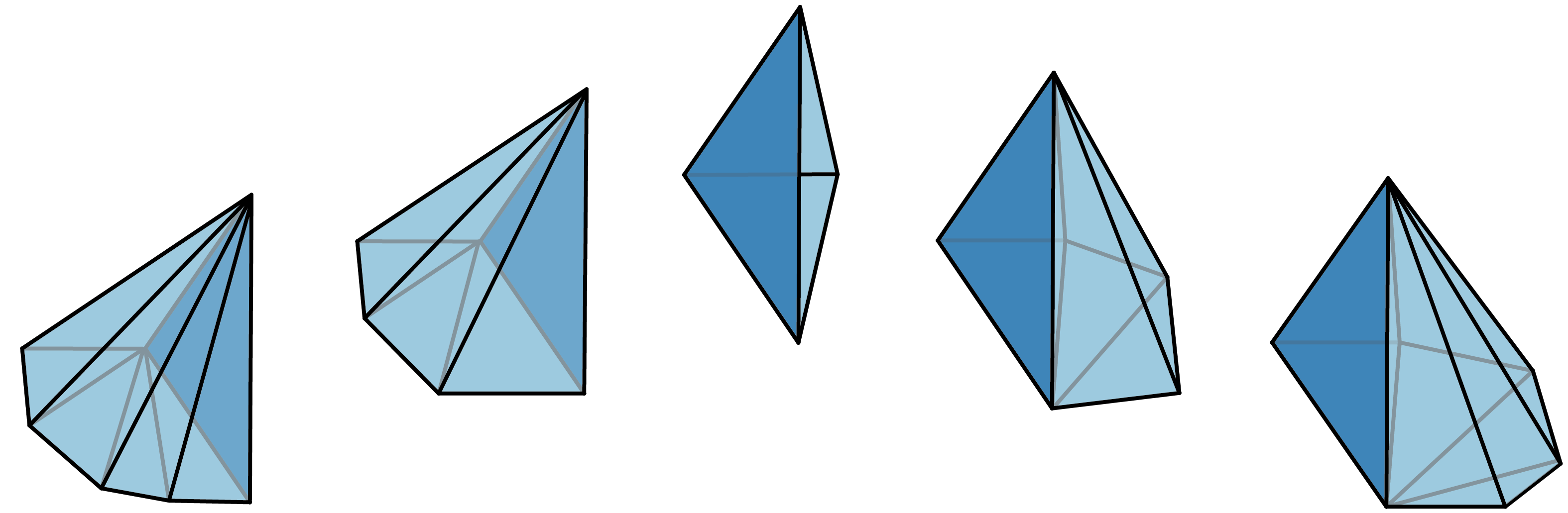}
\caption{$ V_i$ for $i=1,\cdots,5$
}
\label{fig:Vi}
\end{figure}

\section{Collapsing measurement spaces}
\label{sec:collapsing-measurement-spaces}

In this section we   study the effect of collapsing {simplices} in the measurement space. 
{This method {is} very effective in describing the vertices of the polytope of simplicial distributions.
}

Let us begin with the simplest case of collapsing a single edge to a point. 
Recall that $\Delta^1$ is the simplicial set representing an edge. It has a single generating simplex in dimension $1$ denoted by $\sigma^{01}$. A point is represented by the simplicial set $\Delta^0$. Its $n$-simplices are given by $c_n$
obtained by applying the $s_0$ degeneracy map  $n$-times: $s_0\cdots s_0(\sigma^0)=\sigma^{0\cdots 0}$.
Collapsing an edge to a point can be represented by a simplicial set map 
$$
\pi:\Delta^1 \to \Delta^0
$$
that sends
the generating simplex $\sigma^{01}$ to the degenerate simplex $\sigma^{00}$.
Now, applying the cone construction to this map we obtain a simplicial set map
$$\Cone\pi:\Cone(\Delta^1) \to \Cone(\Delta^0)$$
{Recall from Fig.~(\ref{fig:cone-delta1})} that $\Cone(\Delta^1)$ can be identified with a triangle whose generating simplex is given by $(c,\sigma^{01})$.
A similar topological intuition works for $\Cone(\Delta^0)$. It represents an edge whose generating simplex is $(c,\sigma^0)$. From this we can work out the map $\Cone\pi$ as follows: The generating simplex $(c,\sigma^{01})$ is mapped to $(c,s_0\sigma^0)$ since $\pi$ sends $\sigma^{01}$ to the degenerate simplex $s_0\sigma^0$. By the simplicial structure of the cone described in Definition \ref{def:cone} we have
\begin{equation}\label{eq:s1}
(c,s_0\sigma^0) = s_1 (c,\sigma_0).
\end{equation}
As we have seen in Section \ref{sec:gluing-extending} the map $\Cone\pi$ between the cone spaces induces a map between the associated simplicial distributions
$$
(\Cone\pi)^*:\sDist(\Cone(\Delta^0))\to \sDist(\Cone(\Delta^1))
$$
A simplicial distribution $p\in \sDist(\Cone(\Delta^0))$ is determined by $p_{(c,\sigma^0)}\in D_{R}(\ZZ_2)$. 
Let $q$ denote the image of $p$,  i.e., $q=(\Cone\pi)^*(p)$.
Then $q$ will be determined by $q_{(c,\sigma^{01})}$, a distribution on $\ZZ_2^2$.
It is given as follows; see Figure (\ref{CollDistPic}):
\begin{equation}\label{eq:collDist}
\begin{aligned}
q_{(c,\sigma^{01})}^{ab} &= p_{\Cone{\pi}(c,\sigma^{01})}^{ab} \\
&=p_{(c,s_0\sigma^0)}^{ab}\\
& = p_{s_1(c,\sigma^0)}^{ab} \\
& = \left(D_{R}(s_1)(p_{(c,\sigma^0)})\right)^{ab}\\
& = \left\lbrace
\begin{array}{ll}
p^0_{(c,\sigma^0)} & (a,b)=(0,0) \\
1-p^0_{(c,\sigma^0)} & (a,b)=(1,0) \\
0                    & \text{otherwise,}
\end{array}
\right.
\end{aligned}
\end{equation}
{where in the first line we use $q=p\circ \Cone \pi$, in the second line the definition of $\Cone\pi$, in the third line Eq.~(\ref{eq:s1}), {in the fourth line compatibility of $p$ with the simplicial structure, and in the fifth line} 
the definition of $D_{R}(s_1)$.}
We will refer to this distribution as a \emph{collapsed distribution} on the triangle.

\begin{figure}[h!]
\centering
  \includegraphics[width=.25\linewidth]{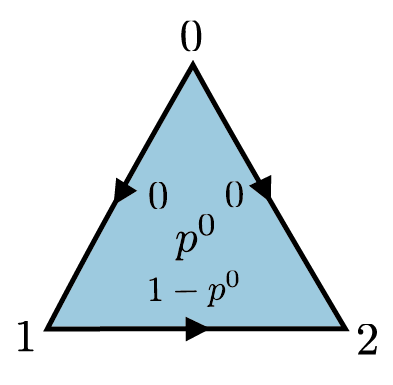}
\caption{{Collapsed distribution on the triangle.}
}
\label{CollDistPic}
\end{figure}

Next we consider the general case. Let $X$ be a $1$-dimensional simplicial set and $\sigma$ denote a non-degenerate $1$-simplex such that
\begin{equation}\label{eq:d0-neq-d1}
d_0(\sigma) \neq d_1(\sigma). 
\end{equation} 
We will write $X/\sigma$ for the simplicial set obtained by collapsing\footnote{The simplicial set $X/\sigma$ is the quotient of $X$ by the simplicial subset generated by $\sigma$ in the sense of \cite[Section 5.3]{okay2022simplicial}.} this edge. 
More formally, $X/\sigma$ consists of the same generating simplices as $X$ except $\sigma$ and the simplicial relations are inherited from $X$. We will write $\pi:X\to X/\sigma$ for the collapsing map as before.
We also have $\Cone\pi:\Cone(X)\to \Cone(X/\sigma)$ {that collapses the triangle obtained as the cone of $\sigma$}.
In Fig.~(\ref{fig:PicA0}) we represent the collapsing map $\pi:C_4\to C_3$ between two {circle} scenarios. 
%
%
%
\begin{lem}\label{Lem:ColDis}
For a collapsing map $\pi:X \to X/\sigma$, the following properties hold.
\begin{enumerate}
\item The map 
$$(\Cone\pi)^*: \sDist(\Cone(X/\sigma))\to \sDist(\Cone(X))$$ 
is injective. 
Moreover, a distribution $q\in \sDist(\Cone(X))$ lies in the image of $(\Cone\pi)^*$ if and only if $q_{(c,\sigma)}$ is a collapsed distribution.  
\item The map 
$$(\Cone\pi)^*:\dDist(\Cone(X/\sigma)) \to \dDist(\Cone(X))$$ 
is injective. Moreover, a deterministic distribution $\delta^t \in\dDist(\Cone(X))$ lies in the image of  $(\Cone \pi)^\ast$ if and only if $t_{(c,\sigma)}\in \{(0,0),(1,0)\}$.
\end{enumerate}  
\end{lem}
\begin{proof} 
The surjectivity of $(\Cone \pi)_n$ for every $n\geq 0$ implies the injectivity of $(\Cone \pi)^\ast$ {in both cases}.  
For  $p \in \sDist(\Cone(X))$ the definition of the collapsing map implies that for every generating $1$-simplex $\tau \neq \sigma$ in $X$ we have $(\Cone{\pi})^\ast(p)_{(c,\tau)}
=p_{(c,\tau)}$. Using Eq.~(\ref{eq:collDist}) we obtain that $(C\pi)^\ast(p)_{(c,\sigma)}$ is a collapsed distribution.
By part (1), a deterministic distribution $\delta^t \in\dDist(\Cone(X))$ lies in the image of  $(\Cone \pi)^\ast$ if and only if 
$\delta^{t_{(c,\sigma)}}$ is a collapsed distribution. This is equivalent to  $t_{(c,\sigma)}\in \{(0,0),(1,0)\}$.
\end{proof}
%
%
%
%
%

\begin{thm}\label{thm:collapsing}
Let $X$ be a $1$-dimensional simplicial set, and $\pi: X \to X/\sigma$ denote a collapsing
map. 
For $p \in \sDist(\Cone(X/\sigma))$ {and $q=(\Cone\pi)^*(p)$}, the following holds. 
\begin{enumerate}
\item $p$ is contextual if and only if $q$ is contextual.
\item $p$ is strongly contextual  if and only if $q$ is strongly contextual.
\item $p$ is a vertex if and only if $q$ is a vertex.
\item $p$ is deterministic distribution if and only if $q$ is deterministic distribution.
\end{enumerate} 
\end{thm}
\begin{proof}
Part (1): Proposition \ref{pro:pro4-1} implies that if $q$ is contextual then $p$ is contextual.
For the converse assume that $q$ is non-contextual. Then there exists {$d=\sum_{i=1}^n d(s^i)\delta^{s^i}\in
  D_{R}\left(\dDist(\Cone (X))\right)$, where $d(s^i)\neq 0$,} such that for every simplex $\tau$ in $\Cone(X)$ we have 
\begin{equation}\label{eq:CpiPnoncon}
q_{\tau}= \sum_{i=1}^n d(s^i)\delta^{s^i_\tau}.
\end{equation}
By part (1) of Lemma \ref{Lem:ColDis} the distribution $q_{(c, \sigma)}$ is collapsed. By Eq.~(\ref{eq:CpiPnoncon}) we conclude that $s^i_{(c,\sigma)}\in \{(0,0),(1,0)\}$ 
{for every $i=1,\cdots,n$.}
Therefore, by part (2) of  Lemma \ref{Lem:ColDis} we have $\delta^{r^i} \in \dDist(\Cone(X/\sigma))$ such that 
$(\Cone \pi)^\ast(\delta^{r^i})
=\delta^{s^i}$. 
We define $\tilde{d} \in D_{R}\left(\sDist(\Cone(X/\sigma))\right)$ by  $\sum_{i=1}^n d(r^i)\delta^{r^i}$.
Then $D_{R}((\Cone\pi)^\ast)(\tilde{d})=d$. Therefore, using the commutativity of Diagram (\ref{dia:Theta-f*}) for $f=\Cone \pi$, we obtain that
$$
(\Cone\pi)^\ast
(\Theta(\tilde{d}))=
\Theta(D_{R}\left((\Cone \pi)^\ast\right)(\tilde{d}))
=\Theta(d)=(\Cone \pi)^\ast(p)
$$
Since $(\Cone \pi)^\ast$ is injective $\Theta(\tilde{d})=p$, which means that $p$ is non-contextual.

Part (2): If $q$ is strongly contextual then $p$ is strongly contextual by \cite[Lemma 5.19, {part (1)}]{kharoof2022simplicial}.
For 
the converse
assume that $s \in \supp(q)$. Then $q_{(c,\sigma)}\left(s{(c,\sigma)}\right) \neq 0$. By part (1) of Lemma \ref{Lem:ColDis} the distribution $q_{(c,\sigma)}$ is collapsed. We conclude that $s_{(c,\sigma)} \in \{(0,0),(1,0)\}$. By part (2) of Lemma \ref{Lem:ColDis} there exists $\delta^{r} \in \dDist(X/\sigma)$ such that $(\Cone \pi)^\ast(\delta^r)=\delta^s$.
To show that $r \in \supp(p)$ it is enough to prove that for every non-degenerate simplex $\tau \in (X/\sigma)_1$ we have $p_{(c,\tau)}(r_{(c,\tau)})\neq 0$. Note that since $\tau \neq \sigma$ we have
$$
p_{(c,\tau)}=q_{(c,\tau)} \;\;  \text{and} \;\; \delta^r_{(c,\tau)}=(\Cone \pi)^\ast(\delta^r)_{(c,\tau)}
=\delta^s_{(c,\tau)}.
$$
Therefore $p_{(c,\tau)}(r_{(c,\tau)})\neq 0$ since $s \in \supp(q)$.

Part (3):
According to \cite[Corollary 5.16]{kharoof2022simplicial}  every vertex in {the preimage of $q$ under $(\Cone\pi)^*$}
is a vertex in $\sDist(\Cone X/\sigma)$. Because of the injectivity of $(\Cone\pi)^\ast$ 
{this preimage} 
contains just $p$, thus $p$ is a vertex.
For the converse, suppose we have distributions $q_1,q_2 \in \sDist(\Cone X)$ and $0<\alpha<1$ such that 
$$ 
q=
\alpha q^1 + (1-\alpha)q^2 
$$
By part (1) of Lemma \ref{Lem:ColDis} the distribution $q_{(c,\sigma)}$ is a collapsed distribution, thus $q^1_{(c,\sigma)}$ and $q^2_{(c,\sigma)}$ are also collapsed. 
Again, by part (1) of Lemma \ref{Lem:ColDis} there exists $\tilde{q}^1, \tilde{q}^2  \in 
\sDist(\Cone( X/\sigma))$ such that $(\Cone \pi)^\ast(\tilde{q}^i)=q^i$, which implies that 
$$
q=
\alpha (\Cone \pi)^\ast(\tilde{q}^1)+(1-\alpha) (\Cone \pi)^\ast(\tilde{q}^2)
=(\Cone \pi)^\ast(\alpha \tilde{q}^1 + (1-\alpha)\tilde{q}^2)
$$
The map $(\Cone \pi)^\ast$ is injective, therefore $p=\alpha \tilde{q}^1 + (1-\alpha)\tilde{q}^2$. Since $p$ is a vertex and $0<\alpha<1$, we conclude that $\tilde{q}^1=\tilde{q}^2$. Therefore $q^1=q^2$.
Part (4): By \cite[Proposition 5.14]{kharoof2022simplicial} 
every deterministic distribution is a vertex in the polytope of simplicial distributions. Thus we can characterize deterministic distributions as the only non-contextual vertices. Then we obtain this result from parts (1) and (3).
\end{proof}

\begin{cor}\label{Cor: Collap}
A distribution $p\in \sDist(\Cone(X/\sigma))$ is a contextual vertex if and only if $(\Cone\pi)^*(p)$ is a contextual vertex. 
\end{cor}
\Proof{
This follows directly  from parts (1) and (3) of Theorem \ref{thm:collapsing}.
}

{\rm
\rem{\label{rem:CollapT}
{
{\rm
The conclusions of Theorem \ref{thm:collapsing} and Corollary \ref{Cor: Collap} hold for more general kinds of collapsing maps obtained by collapsing a set of edges in sequence.
}
}
}}

\subsection{Application to Bell inequalities}
\label{sec:application-vell-ineq}

Consider the collapsing map $\pi:X\to X/\sigma$ and suppose that the Bell inequalities for the scenario 
$\Cone(X)$ are known.
By part (1) of Theorem \ref{thm:collapsing} a simplicial distribution $p \in \sDist(\Cone(X/\sigma))$ is non-contextual if and only if 
$q=(\Cone\pi)^\ast(p) \in \sDist(\Cone(X))$ is non-contextual.
This is equivalent to the condition that $q$ satisfies the Bell inequalities for the scenario $\Cone(X)$.
From these Bell inequalities we can extract those for the collapsed scenario $\Cone(X/\sigma)$.

Let us illustrate how the collapsing technique can be applied to cycle scenarios. 
Let $\pi: C_4 \to C_3$ denote the map that collapses one of the edges in the $4$-{circle} space as in Fig.~(\ref{fig:PicA0}).
\begin{figure}[h!]
\centering
\begin{subfigure}{.49\textwidth}
  \centering
  \includegraphics[width=.6\linewidth]{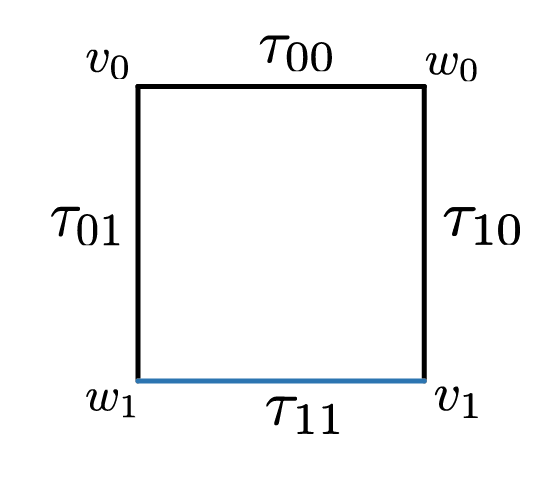}
  \caption{}
  \label{fig:PicA0-a}
\end{subfigure}%
\begin{subfigure}{.49\textwidth}
  \centering
  \includegraphics[width=.64\linewidth]{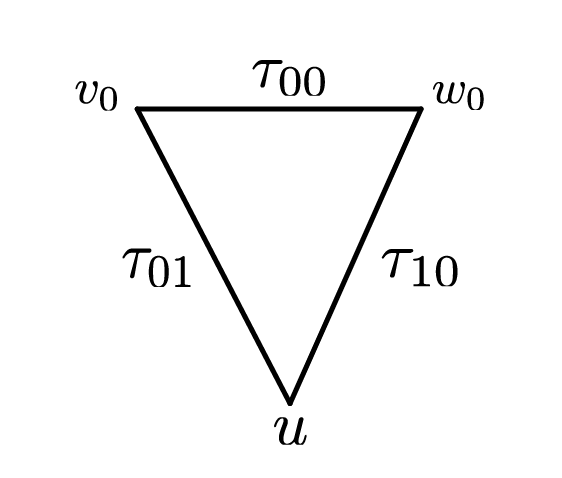}
  \caption{}
  \label{fig:PicA0-b}
\end{subfigure}
\caption{The edge $\tau_{11}$ (blue) is collapsed.
}
\label{fig:PicA0}
\end{figure} 
By Proposition \ref{pro:fConvex} a simplicial distribution $p\in \sDist(\tilde{C}_3)$ is specified by a tuple
$$
p=(p_{(c,v_0)},p_{\tau_{00}}
,p_{(c,w_0)},
p_{\tau_{10}},p_{(c,u)},
p_{\tau_{01}})
$$ 
where each entry is a distribution on $\ZZ_2$.
On the other hand, a simplicial distribution $q$ on $\tilde C_4$ is specified by a tuple 
$$(q_{(c,v_0)},q_{\tau_{00}}
,q_{(c,w_0)},
q_{\tau_{10}},q_{(c,v_1)},
q_{\tau_{11}}
,q_{(c,w_1)},
q_{\tau_{01}})
$$
Then the image of $p$ under the map $(\Cone\pi)^*$ gives us  $q=(p_{(c,v_0)},p_{\tau_{00}},
p_{(c,w_0)},
p_{\tau_{10}},p_{(c,u)},1,
p_{(c,u)},{p_{\tau_{01}}})$.
This latter simplicial distribution is non-contextual if and only if it satisfies the $4$-{circle inequalities} 
(see Eq.~(\ref{eq:CHSH-ineq-edge})):
$$
\begin{aligned}
0\leq p_{\tau_{00}}+
p_{\tau_{10}}+1-p_{\tau_{01}} \leq 2 \\
0\leq p_{\tau_{00}}+
p_{\tau_{10}}-1+p_{\tau_{01}} \leq 2 \\
0\leq p_{\tau_{00}}-
p_{\tau_{10}}+1+p_{\tau_{01}} \leq 2\\
0\leq -p_{\tau_{00}}+
p_{\tau_{10}}+1+p_{\tau_{01}} \leq 2.
\end{aligned}
$$
Half of these inequalities are trivial, so this set of inequalities is equivalent to  
$$
\begin{aligned}
p_{\tau_{00}}+
p_{\tau_{10}}-p_{\tau_{01}} \leq 1 \\ p_{\tau_{00}}+
p_{\tau_{10}}+p_{\tau_{01}} \geq 1 \\
p_{\tau_{00}}-
p_{\tau_{10}}+p_{\tau_{01}} \leq 1\\
-p_{\tau_{00}}+
p_{\tau_{10}}+p_{\tau_{01}} \leq 1
\end{aligned}
$$  
which constitute the non-trivial $3$-{circle inequalities}.

\begin{figure}[h!]
\centering
\begin{subfigure}{.49\textwidth}
  \centering
  \includegraphics[width=.77\linewidth]{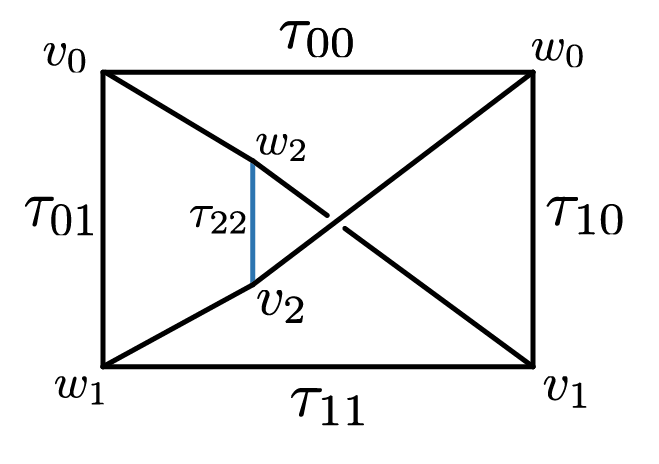}
  \caption{}
  \label{fig:PicA1}
\end{subfigure}%
\begin{subfigure}{.49\textwidth}
  \centering
  \includegraphics[width=.81\linewidth]{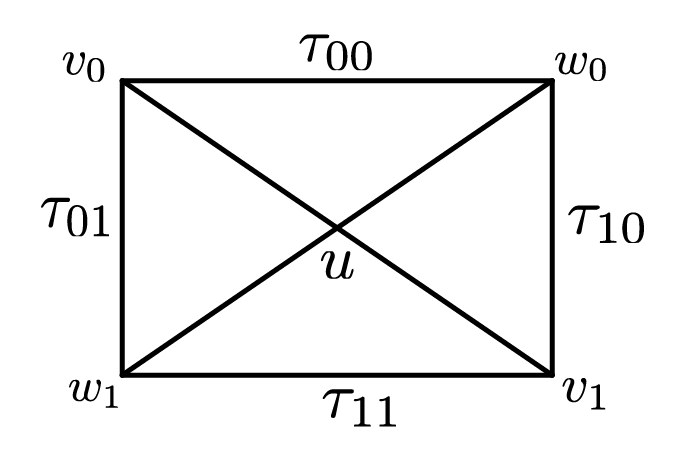}
  \caption{}
  \label{fig:PicA2}
\end{subfigure}
\caption{ 
(a) The complete bipartite graph $K_{3,3}$. (b) The graph obtained by collapsing the edge $\tau_{22}$ to the vertex $u$.
}
\label{fig:collapse-K33}
\end{figure} 

Now, we will apply this technique to find the Bell inequalities for the cone of the  $1$-dimensional space given in Fig.~(\ref{fig:PicA2}).
%
%
This space will be our collapsed space $X/\sigma$.
The $1$-dimensional simplicial set $X$ is the complete bipartite graph $K_{3,3}$ given in Fig.~(\ref{fig:PicA1}).
We denote the edge from $v_i$ to $w_j$ by $\tau_{ij}$.
{Note that the measurement space $\Cone(X)$ represents the $(3,3,2,2)$ Bell scenario.
It }has three kinds of Bell inequalities: (1) trivial, (2) {circle inequalities}, and (3) Froissart inequalities \cite{froissart1981constructive}. 
We are interested in the latter type. An example of Froissart inequalities is the following (see \cite[Eq.~(21)]{brunner2014bell}):  
$$
p_{(c,v_0)}^0 +p_{(c,w_0)}^0-p_{(c,\tau_{00})}^{00}-p_{(c,\tau_{01})}^{00}-p_{(c,\tau_{02})}^{00}-p_{(c,\tau_{10})}^{00}-p_{(c,\tau_{20})}^{00}-p_{(c,\tau_{11})}^{00}+p_{(c,\tau_{12})}^{00}+p_{(c,\tau_{21})}^{00} \geq -1
$$  
It will be convenient for us to convert this inequality to one that only contains distribution on edges. For this we will use 
Eq.~(\ref{eq:p-ab}).
This substitution gives us the following inequality:
\begin{equation}\label{equ:FroBellineq}
\begin{aligned}
p^0_{\tau_{12}}
+p^0_{\tau_{21}}
-p^0_{\tau_{00}}
 - p^0_{\tau_{01}}
&-p^0_{\tau_{02}}
-p^0_{\tau_{10}}
-p^0_{\tau_{20}}
 - p^0_{\tau_{11}}
\\
&-p_{(c,v_0)}^0
-p_{(c,w_0)}^0 
-p_{(c,v_1)}^0 
-p_{(c,w_1)}^0 
 \geq -6
\end{aligned}
\end{equation}
We observe that the only edges that do not appear in this inequality are $(c,v_2),(c,w_2),\tau_{22}$. Applying the symmetries of $X$, more precisely the automorphism group of the graph $K_{3,3}$, we can obtain $9$ distinct such inequalities in which the edges $(c,v_i),(c,w_j),\tau_{ij}$ do not appear where $i,j\in\{0,1,2\}$.
For example, for {$i=1, j=2$} we have
\begin{equation}\label{equ:FroBellineq2}
\begin{aligned}
p^0_{\tau_{02}}
+p^0_{\tau_{10}}
-p^0_{\tau_{00}}
 - p^0_{\tau_{01}}
&-p^0_{\tau_{22}}
-p^0_{\tau_{21}}
-p^0_{\tau_{20}}
 - p^0_{\tau_{11}}
\\
&-p_{(c,v_0)}^0
-p_{(c,w_0)}^0 
-p_{(c,v_2)}^0 
-p_{(c,w_1)}^0 
 \geq -6
\end{aligned}
\end{equation}
Note that these $9$ inequalities are in distinct orbits under the action of $\dDist(\Cone(X))$
since different edges appear in every one of them (see Example \ref{ex:chsh}).
We can find the number of the Froissart inequalities in every orbit. 
{A deterministic distribution $\delta^s$ fixes the inequality (\ref{equ:FroBellineq}) if and only if $(\delta^s_{\tau})^{0}=1$ for every edge $\tau$ that appears in the inequality. In this case, $\delta^s$ is the identity, i.e., $(\delta^s_\sigma)^{00}=1$ for every triangle $\sigma$ in $\Cone(X)$.} 
This implies that the size of the orbit of this inequality is equal to $|\dDist(\Cone(X))|=2^6=64$. Same counting argument works for the rest of the $9$ inequalities. Therefore there are $9 \cdot 64=576$ Bell-inequalities of this type.

Next, we apply our collapsing technique to generate a new Bell inequality, i.e., one that is not a {circle inequality} for the cone of the scenario given in  Fig.~(\ref{fig:PicA2}).
We will use the following  collapsing map $\pi:X\to X/\sigma$ where $\sigma=\tau_{22}$.
For a given {circle} on $X$ there is a corresponding {circle inequality}.
Every such inequality will appear as a Bell inequality in the collapsed scenario $X/\sigma$ if the corresponding {{circle}} does not contain the collapsed edge $\tau_{22}$. 
If the {{circle}} contains $\tau_{22}$ the resulting Bell inequality will be a {circle inequality} of size one less as in Definition \ref{def:cycle-ineq}.
Given a simplicial distribution $p\in \sDist(X/\sigma)$ the image $q=(\Cone\pi)^*(p)$ satisfies
$$
q_{\tau_{22}}^0=1\;\text{ and }\; q_{(c,w_2)}^0=q_{(c,v_2)}^0=p_{(c,u)}^0. 
$$ 
Substituting this in {(\ref{equ:FroBellineq2}) } we obtain the following Bell inequality of the scenario $X/\sigma$:
\begin{equation}\label{equ:NewBellineq}
\begin{aligned}
p^0_{\tau_{02}}
+p^0_{\tau_{10}}
-p^0_{\tau_{00}}
 - p^0_{\tau_{01}}
-p^0_{\tau_{21}}
-p^0_{\tau_{20}}
 - p^0_{\tau_{11}}
 -p_{(c,v_0)}^0
-p_{(c,w_0)}^0 
-p_{(c,u)}^0 
-p_{(c,w_1)}^0 
 \geq -5
\end{aligned}
\end{equation}
This inequality is a new Bell inequality, that is, it is not a {circle inequality}, and it belongs to a scenario that is not a Bell scenario.
The latter observation implies that going beyond Bell scenarios can produce simpler Bell inequalities that are not {circle inequalities}; see \cite{pironio2014all}. 

\Rem{\label{rem:conjecture}
{\rm
Let $X$ be a $1$-dimensional simplicial set. Consider a  Bell inequality for the cone scenario $\Cone(X)$ expressed in the edge coordinates (see Proposition \ref{pro:fConvex}).
Then in the known examples the edges that appear with non-trivial coefficients in this Bell inequality form a loop {(i.e., a {circle} with possible self intersections)} on $X$. It is a curious question whether this observation holds for every $1$-dimensional $X$. If so, it gives a topological restriction on the form of possible Bell inequalities hence a nice structural result in contextuality.  
}
}

\subsection{Detecting contextual vertices} 

In this section, the $1$-circle $C_1$ will play a fundamental role to detect contextual vertices in the scenarios of interest.
Let $\tau$ denote the generating $1$-simplex of $C_1$. 
A simplicial distribution $p\in \sDist(\tilde C_1)$
is specified by 
$$(p_{(c,\tau)}^{00},p_{(c,\tau)}^{01},p_{(c,\tau)}^{10},p_{(c,\tau)}^{11})$$
where 
$$p_{(c,\tau)}^{ab} \in [0,1],\;\;\sum_{a,b}p_{(c,\tau)}^{ab}=1 \;\;\text{ and }\;\; 
p_{(c,\tau)}^{00}+p_{(c,\tau)}^{01}=p_{(c,\tau)}^{00}+p_{(c,\tau)}^{11}.$$
This implies that $p_{(c,\tau)}^{01}=p_{(c,\tau)}^{11}$. Therefore
the polytope $\sDist(\tilde C_1) \subset \RR^3$ 
is a triangle with two deterministic vertices and {a unique} contextual vertex $p_{-}$ given in Fig.~(\ref{fig:D-S-vert1}).

\begin{figure}[h!]
\centering
\begin{subfigure}{.66\textwidth}
  \centering
  \includegraphics[width=.3\linewidth]{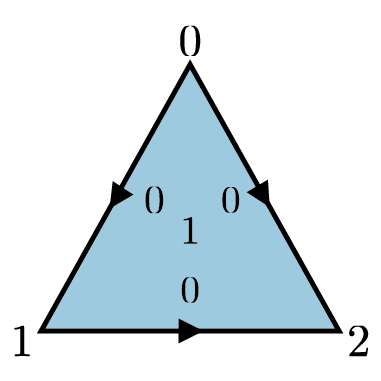}\;\;\;\;\;
    \includegraphics[width=.3\linewidth]{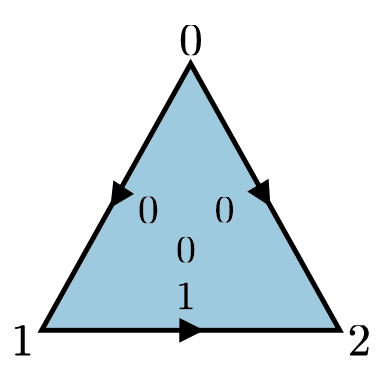}
  \caption{}
  \label{fig:SimpDetVertex}
\end{subfigure}%
\begin{subfigure}{.33\textwidth}
  \centering
  \includegraphics[width=.6\linewidth]{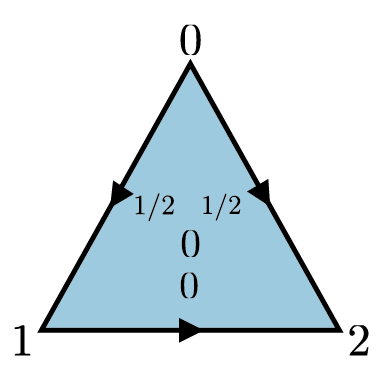}
  \caption{}
  \label{fig:SVertPic}
\end{subfigure}
\caption{{The $1$-cycle scenario is obtained by identifying the {edges $(0,1)$ and $(1,2)$} of a triangle.}
Deterministic (a) and contextual (b) vertices  of the $1$-cycle scenario.
}
\label{fig:D-S-vert1}
\end{figure}

The following example shows how to obtain a contextual vertex in an arbitrary cycle scenario from the contextual vertex in Fig.~(\ref{fig:SVertPic}) using  the collapsing technique.

\Ex{\label{ex:DetectPR}
{\rm
Let $\pi:C_N \to C_1$ denote the map that collapses $\tau_2,\cdots,\tau_N$ in the $N$-{circle} scenario.
We will write $\tau=\tau_1$ for notational simplicity.
Let $q=(\Cone\pi)^*(p_-)$ where $p-$ is the unique contextual vertex of $\sDist({\tilde C_1})$.
{We have}
$$
q_{(c,
d_i(\tau))}^0=(p_{-})_{(c,
d_i(\tau))}^0=1/2$$
(see Fig.~(\ref{fig:SVertPic})). Then {by Eq.~(\ref{eq:collDist})} for $\tau'=\tau_2$ and $\tau_N$ we have
$$q_{(c,
\tau')}=p_{+}.$$
We can continue this way using Eq.~(\ref{eq:collDist}) to obtain that for an edge $\tau'$ of the $N$-cycle scenario we have
$$
q_{(c,\tau')}=\left\lbrace
\begin{array}{ll}
p_{-} & \tau'=\tau \\
p_{+} & \text{otherwise.}
\end{array}
\right.
$$
According to Corollary \ref{Cor: Collap}, $q$ is a contextual vertex in the $N$-cycle scenario; see Fig.~(\ref{fig:CnVerPic}). 
}
}

\begin{figure}[h!]
\centering
  \includegraphics[width=.35\linewidth]{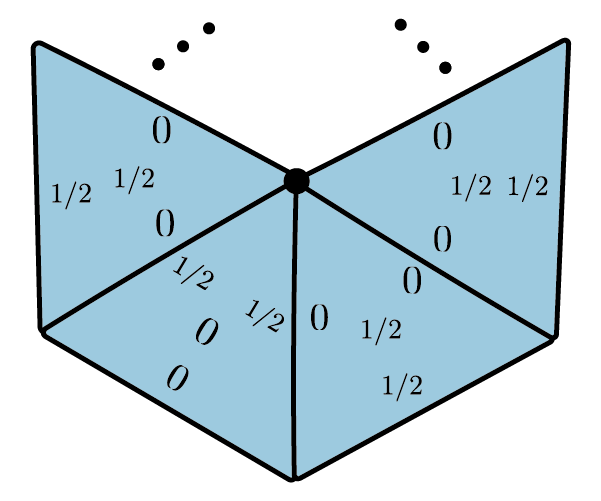}
\caption{ 
} 
\label{fig:CnVerPic}
\end{figure}

The vertex detected in Example \ref{ex:DetectPR}   generalizes the PR boxes defined in Example \ref{ex:chsh}. 
 
\Def{\label{def:PR-box}
A {\it PR box} {on an $N$-cycle scenario} is a simplicial distribution $p\in \sDist(\tilde C_N)$ such that $p_{\sigma}=p_\pm$ for $\sigma\in \set{\sigma_{i}:\,i=1\cdots,N}$ with the further restriction that the number of $p_-$'s is odd.
}

It is clear from the definition that there are $2^{N-1}$ PR boxes on the $N$-cycle scenario. All of them can be obtained from the one given in Example \ref{ex:DetectPR} by the action of $\dDist(\tilde C_N)$ in a way similar to the CHSH scenario discussed in Example \ref{ex:chsh}. Therefore by Proposition \ref{pro:action} the PR boxes are contextual vertices in the cycle scenario.

{Next we describe a $1$-dimensional simplicial set obtained by gluing $n$ copies of $C_1$ at their vertex.
More explicitly, this simplicial set, denoted by $\vee_{i=1}^n C_1$,  consists of the generating $1$-simplices $\tau_1,\cdots,\tau_n$ with the identifying relations
$$
d_0\tau_i =v= d_1\tau_i,\;\;\; i=1,\cdots,n,
$$ 
where $v$ is the unique vertex. 
The cone space $\Cone(\vee_{i=1}^n C_1)$ consists of $n$ triangles given by $(c,\tau_i)$ where $i=1,\cdots,n$.
}
We  consider simplicial distributions on the cone of $\vee_{i=1}^n C_1$.
Such a distribution $p$ is determined by the $n$ distributions $p_{(c,\tau_i)}$ on $\ZZ_2^2$.
For convenience of notation we will write $(p_1,\cdots,p_n)$ for this tuple of distribution, i.e., $p_i=p_{(c,\tau_i)}$.

\begin{pro}\label{pro:ConeofCircCube}
The polytope $\sDist(\Cone(\vee _{i=1}^n C_1))$ 
can be identified with
the following subpolytope of the $(n+1)$-cube
\begin{equation}\label{eq:subpolytope}
\{(a_1,\dots,a_n,b) \in [0,1]^n:\,\frac{1-a_i}{2} \leq b \leq \frac{1+a_i}{2},\; \forall 1\leq i \leq n \}
\end{equation}
\end{pro}
\begin{proof}
The non-degenerate edges in 
$\Cone(\vee _{i=1}^n C_1)$ are $\tau_1,\dots,\tau_n,(c,v)$, and for every $1\leq i \leq n$ the non-degenerate triangle $(c,\tau_i)$ has the edges $(c,v),\tau_i,(c,v)$. By Proposition \ref{pro:fConvex} and Corollary  \ref{cor:image-f} we see that $\sDist(\Cone(\vee _{i=1}^n C_1))$ can be identified with  the set of $(a_1,\dots,a_n,b) \in [0,1]^n$ satisfying the following inequalities: 
%
%
\begin{equation}\label{eq:Rest}
 \begin{array}{rcl}
b + a_i +  b \geq 1 \\ 
b + a_i -  b \leq 1  \\
b - a_i +  b \leq 1 \\
-b + a_i +  b \leq 1 
\end{array} 
\end{equation}
%
%
for every $1\leq i \leq n$. 
The set of inequalities in (\ref{eq:Rest}) is equivalent to $ \frac{1-a_i}{2} \leq b \leq \frac{1+a_i}{2}$, $1\leq i\leq n$.
\end{proof}

\begin{pro}\label{pr: VerWedge}
The polytope $\sDist(\Cone(\vee _{i=1}^n C_1))$ has $2^{n}+1$ vertices: 
\begin{enumerate}
\item There are two deterministic vertices given by 
$(\delta^{00},\cdots,\delta^{00})$ and $(\delta^{11},\cdots,\delta^{11})$.

\item The contextual vertices are of the form $(p_1,\cdots,p_n)$ where $p_i\in \{p_+,p_-\}$ for every $1 \leq i \leq n$   with at least one $j$ satisfying $p_j=p_-$.
\end{enumerate} 
\end{pro}
\begin{proof}
The edge $(c,v)$ appears twice in every non-degenerate triangle of $\Cone(\vee _{i=1}^n C_1))$, 
thus every outcome assignment $s$ on this measurement space
is determined by $s_{(c,v)}\in \set{0,1}$.
Therefore we have only 
$(\delta^{00},\cdots,\delta^{00})$ and $(\delta^{11},\cdots,\delta^{11})$ as deterministic distributions.
Now, let us denote the polytope in (\ref{eq:subpolytope}) by $P$ and find its vertices. Given an element $
(a_1,\dots,a_n,b)\in P$ such that $a_j\notin \{0,1\}$ for some $1\leq j \leq n$, there exists distinct $a'_j,a''_j \in [0,1]$, and $0 <\alpha<1$ such that 
$$
\frac{1-a'_j}{2} \leq q \leq \frac{1+a'_j}{2} , \;\; 
\frac{1-a''_j}{2} \leq q \leq \frac{1+a''_j}{2} , 
\;\; \text{and} \;\;\alpha a'_j+(1-\alpha)a''_j=a_j 
$$ 
Therefore we have
$$
(a_1,\dots,a_n,q)=\alpha(a_1,\dots,a'_j,\dots,a_n,q)+(1-\alpha)(a_1,\dots,a''_j,\dots,a_n,q)
$$
We conclude that if $
(a_1,\dots,a_n,q)$ is a vertex in $P$, then $a_1,\dots,a_n \in \{0,1\}$. In the case that $a_1=\dots=a_n=1$, we have two vertices $(1,\dots,1,0)$ and $(1,\dots,1,1)$. 
Let $f:\Cone(\vee _{i=1}^n C_1) \to P$ denote the bijection  given in {Proposition} \ref{pro:ConeofCircCube}.
One can see that by applying the inverse of $f$
we obtain the two deterministic vertices 
$(\delta^{00},\cdots,\delta^{00})$ and $(\delta^{11},\cdots,\delta^{11})$. 
On the other hand, if $a_j=0$ for some $0 \leq j \leq 1$, then 
$$
\frac{1}{2}=\frac{1-0}{2} \leq q \leq \frac{1+0}{2}=\frac{1}{2}.
$$
We obtain that $q=\frac{1}{2}$. Therefore the rest of the vertices are of the form $(a_1,\dots,a_n,\frac{1}{2})$ where $a_i\in \{0,1\}$ for every $i$ and $a_j=0$ for at least  one $j$. By applying the inverse of $f$ we obtain the desired contextual vertices. 
\end{proof}

Our main result in this section relates a topological invariant, the fundamental group, to the number of contextual vertices.
Given a $1$-dimensional simplicial set $X$ regarded as a graph consider a maximal tree $T\subset X$.
The collapsing map can be applied to the edges in $T$ to obtain a map $\pi:X\to X/T$ where $X/T$ is of the form $\vee_{i=1}^{n_X} C_1$. The number $n_X$ is a topological invariant of the graph that gives the non-contractible {circle}s. This number is independent of the chosen maximal tree. The fundamental group  $\pi_1(X)$ is defined to be the free group on the set of $n_X$ edges in $X_1^\circ-T_1^\circ$; {see \cite[Section 1.A]{hatcher}.}

\begin{thm}\label{thm:GenrVert}
Let $X$ be a connected $1$-dimensional measurement space and $n_X$ denote the number of generators of the fundamental group $\pi_1(X)$. 
Then there exists at least $(2^{n_X}-1)2^{|X_0|-1}$ contextual vertices in $\sDist(\Cone(X))$.
\end{thm}
\begin{proof}
For simplicity we will write $n=n_X$.
 Let $T$ be a maximal tree in $X$. We have the collapsing map 
$\pi:X \to X/T=\vee_{i=1}^{n} C_1$. 
According to Corollary \ref{Cor: Collap}  applying $(\Cone\pi)^\ast$ to the contextual vertices  described in Proposition \ref{pr: VerWedge} we obtain contextual vertices of $\sDist(\Cone(X))$.
First, we will show that these vertices are in different orbits under the action of $\sDist(\Cone(X))$. 
Given two different contextual vertices $(p_1,\dots,p_n)$ and $(q_1,\dots,q_n)$ of $\sDist(\Cone(\vee _{i=1}^{n_X} C_1))$, 
  there exists $1 \leq j \leq n$ such that $p_j=p_{-}$ and $q_j=p_{+}$. 
  Fix one circle $C$ in $X$ such that the image of $\pi|_{C}$ contains only $\tau_j$ as a non-degenerate $1$-simplex (i.e, $C$ collapsed to the circle generated by {$\tau_j$}).
We have 
$$
(\Cone\pi)^\ast(p_1,\dots,p_n)|_{\Cone(C)} = \Cone(\pi|_C)^\ast(p_j)=\Cone(\pi|_C)^\ast(p_{-})
$$
which is the PR box of Example \ref{ex:DetectPR}.
On the other hand,
 one can see using the same technique of Example \ref{ex:DetectPR} that the restriction of $(\Cone\pi)^\ast(q_1,\dots,q_n)$ to $\Cone(C)$ is the non-contextual distribution $e_{+}$, the identity element of  $G_\pm({\Cone(X)})$ (see Definition \ref{def:G-plus-minus}). 
Therefore these two restrictions are not in the same orbits under the action of $\sDist(\Cone(C))$. 
 We conclude that $(\Cone\pi)^\ast(p_1,\dots,p_n)$ and $(\Cone\pi)^\ast(q_1,\dots,q_n)$ are not in the same orbit under the 
action of $\sDist(\Cone(X))$. 
So far we have proved that there are $2^{n}-1$ contextual vertices in $\sDist(\Cone(X))$ that lie in different orbits. 
Observe that every such vertex has $p_{\pm}$ on every non-degenerate triangle of $\Cone(X)$, thus the only {two outcome assignments 
{that fix this vertex}
are those that} restrict to  $\delta^{00}$ {on every non-degenerate triangle} or $\delta^{11}$ {on every non-degenerate triangle}.
 We conclude that the orbit of such a vertex has $\frac{|\dDist(\Cone(X))|}{2}=\frac{2^{|X_0|}}{2}=2^{|X_0|-1}$ elements. By Proposition  \ref{pro:action} all these distributions are contextual vertices.     
\end{proof}

\Cor{
A simplicial distribution in the group $G_\pm(CX)$ {(Proposition \ref{pro:G-plus-minus})}  
is non-contextual if and only if it belongs to 
the subgroup 
$$
e_{+}\cdot\dDist(\Cone(X))=\set{e_+\cdot \delta^s:\,\delta^s\in \dDist({\Cone(X)})}.
$$
}
\Proof{
 The element $e_+$ is non-contextual since we have
$$
e_+ = \frac{1}{2}\delta^{t} + \frac{1}{2}\delta^r
$$ 
{where $t_{(c, \tau )} = (0, 0)$
and $r_{(c, \tau)} = (1,0)$ for every $\tau \in X_1^\circ$},
and $(e_+\cdot \delta^s)_\sigma = p_{\pm}$ for every 
{non-degenerate} $2$-simplex $\sigma$ of $\Cone(X)$.
Therefore the coset $e_{+}\cdot\dDist(\Cone(X))$ is a subset of $G_\pm(CX)$ and all its elements are non-contextual.
Moreover, since $\sDist(\Cone(X))$ is a commutative monoid and $e_{+} \cdot e_{+}=e_{+}$, the subset $e_{+}\cdot \dDist(\Cone(X))$ is in fact a subgroup of $G_\pm(CX)$. 


 
To conclude that the remaining distributions are all contextual we will use Theorem  \ref{thm:GenrVert}.
For a $1$-dimensional (connected) simplicial set $X$, the Euler characteristic \cite{hatcher} is given by $\chi(X)=1-n_X$. Alternatively, it can be computed using the formula
$$
\chi(X) = |X_0|-|X_1^{\circ}|.
$$
Therefore we have $n_X=|X_1^{\circ}|-|X_0|+1$. 
Using this we find that the number of contextual vertices detected in 
Theorem \ref{thm:GenrVert} is equal to $2^{|X_1^{\circ}|}-2^{|X_0|-1}$. 
On the other hand, we have
$$
|G_\pm(\Cone X)| - |e_+\cdot \dDist(\Cone X)| = 2^{|X_1^{\circ}|}-2^{|X_0|-1}
$$
where we used the fact that the size of the coset is half the size of $\dDist(\Cone(X))$ {since it is the orbit of $e_+$}. 
Therefore the contextual vertices detected in Theorem \ref{thm:GenrVert} are precisely those distribution in $G_\pm(\Cone X) - \left(e_+\cdot \dDist(\Cone X)\right)$.
}

\Ex{
{\rm
The Bell scenario $(m,n,2,2)$ represented by the complete bipartite graph $K_{m,n}$. This graph has $m+n$ vertices and $m\cdot n$ edges. Therefore, by Theorem \ref{thm:GenrVert} we have at least $2^{mn}-2^{m+n-1}$ contextual vertices in the scenario $(m,n,2,2)$.  
}
}

\subsection{Contextual vertices of the cycle scenario}    

    
We conclude this section by showing that PR boxes  constitute all the contextual vertices {in the cycle scenario} using {the} collapsing method. Let us set $X=\tilde C_N$.   
The polytope $P_{X}$ associated to the cycle scenario has dimension $\RR^{2N}$. 
This is a consequence of Proposition \ref{pro:fConvex} since the number of non-degenerate edges of the $N$-cycle space is $2N$.    

\Lem{\label{lem:det-xi-non-contextual}
Let $p$ be a simplicial distribution on the $N$-cycle scenario such that $p_{x_i}$ is deterministic for some $1\leq i\leq N$.
Then $p$ is non-contextual.
}
\Proof{
Assume that $p_{x_i} = \delta^a$ for some $a$. Let $q$ be the deterministic distribution given by $q_{\sigma_j}=\delta^{00}$ for all $1\leq j\leq N$ distinct from $i$ and $q_{\sigma_i} = \delta^{11}$. Since $(q\cdot p)_{x_i}=\delta^0$ and $p$ is non-contextual if and only if $q\cdot p$ is non-contextual, we can assume that $a=0$.
Let $\bar X$ denote the quotient $X$ obtained by collapsing  $x_i$. The resulting space $\bar X$ is a classical $N$-disk. Consider the  map
$$
\pi^*: \sDist(\bar X) \to \sDist(X)
$$ 
induced by the quotient map $\pi:X\to \bar X$. There exists a simplicial distribution $\tilde p$ on the classical $N$-disk such that $\pi^*(\tilde p)=p$. Since every distribution on the $N$-disk is non-contextual $p$ is non-contextual.  
}

\Pro{\label{pro:contextual-ver-Ncycle}
Contextual vertices of the polytope of simplicial distributions on the $N$-cycle scenario are given by the PR boxes.
}
\Proof{
By Lemma \ref{lem:det-xi-non-contextual} for a vertex $p$ there cannot be a deterministic edge on any of the $x_i$'s. By Corollary \ref{cor:vertex} $p$ is a vertex if and only if $\rank(p)=2N$. Therefore there are precisely $N$ deterministic {edges $z_1,\cdots,z_N$, all of which lie on the boundary.}
The distribution $p_{\sigma_i}$, which is given by $p_\pm$, for each triangle 
 has rank $2$.
Let $(A,b)$ be as in Corollary \ref{cor:H-description} so that $P_X = P(A,b)$.
{Let $p$ be the distribution such that $p|_{\sigma_{i}} = p_{\pm}$ for every $\sigma_{i}$ and let $\zZ_{p}$ index the inequalities tight at $p$. There are $2N$ such tight inequalities since there are $N$ non-degenerate simplices $\sigma_{i}$ and each corresponding distribution $p_{\sigma_{i}}$ has two zeros; see Fig.~(\ref{fig:CnVerPic}). Denoting $A[\zZ_{p}] := A_{p}$, we order the columns of $A_{p}$ by $x_1\cdots ,x_N,z_1,\cdots,z_N$. Up to elementary row operations we have}%
%
$$
A_{p} = \left( 
\begin{matrix}
E & \zero \\
\zero & I
\end{matrix}
\right)
$$
Then $\rank(A_{p}) = N+ \rank (E)$. Multiplying each row of $E$ by $-1$ if necessary we can write
$$
E = \left( \begin{matrix}
1 & (-1)^{c_1+1} & 0 & \cdots & 0 \\
0 & 1 &  (-1)^{c_2+1} & \cdots & 0 \\
\vdots   & \vdots & \ddots & \vdots & \vdots  \\
0   & \cdots & 0 & 1 & (-1)^{c_N+1} \\
(- 1)^{c_1+1} & 0 & \cdots & 0 & 1
\end{matrix}
\right)
$$
Let us define $c=\sum_{i=1}^N c_i \mod 2$. Then $\rank(E)=N$ if $c=1$, otherwise $\rank(E)=N-1$. In the former case we have $p_{x_i}^0 =1 - p_{x_i}^0$, which implies that $p_{x_i}^0=1/2$. Hence $p_{\sigma_i}$'s are all given by $p_{\pm}$. The condition that $c=1$ implies that the number of $\sigma_i$'s with $p_{\sigma_i}=p_-$ is odd. 
}

Two vertices $v, v'$ of a full-dimensional polytope $P\subset \RR^d$ are called neighbors if $\rank(A[\zZ_{v}\cap \zZ_{v'}]) = d-1$.


\Cor{
All neighbors of a PR box are deterministic distributions.
}    
\Proof{
A PR box $p$ corresponds to a non-degenerate vertex, meaning that the number $2N$ of tight inequalities is precisely the dimension of the polytope. 
One property of non-degenerate vertices is that, if $\zZ_{p}$ indexes the tight inequalities of a PR box $p$, and $\zZ$ is a set differing from $\zZ_{p}$ by one element, then $p' = A[\zZ]^{-1}b$ is also a vertex so long as it satisfies the remaining inequalities. 
In this case $p$ and $p'$ are neighbors.
A neighbor is obtained by replacing a tight inequality with another, which amounts to replacing one zero with another. 
Doing so will make one of $x_i$'s a  deterministic edge. Lemma \ref{lem:det-xi-non-contextual} implies that $p'$, if it is a vertex of $P_X$, is non-contextual, thus a deterministic distribution. 
}

  
\subsection{Conclusion}
 
{In this paper, we demonstrate novel techniques from the theory of simplicial distributions introduced in \cite{okay2022simplicial}.
We present topological proofs for the sufficiency of the {circle inequalities} for the non-contextuality of distributions on the cycle scenario.
This proof extends the topological proof of the CHSH scenario in \cite{okay2022simplicial}.
We go beyond the cycle scenarios and study the flower scenario depicted in Fig.(\ref{fig:flower}) that generalizes the bipartite Bell scenarios consisting of $2$ measurements for Alice and $m$ measurements for Bob.
Our main insight in the proof is the topological interpretation of Fourier--Motzkin elimination and the gluing and extension methods of distributions on spaces. 
We also explore two new features of scenarios available in the simplicial setting: (1) Collapsing measurement spaces to detect contextual vertices and (2) Applying the monoid structure of simplicial distributions to generate vertices.
}%
{An appealing feature of the 
collapsing technique
featured here is that previously unknown types of Bell inequalities can be discovered from those that are known; see Section~\ref{sec:application-vell-ineq}. These Bell inequalities may have desirable properties, such as having quantum violations that are more robust to noise, which may be of both theoretical and practical interest.
}



  
\paragraph{Acknowledgments.}
This work is supported by the US Air Force Office of Scientific Research under award
number FA9550-21-1-0002.  

\bibliography{bib.bib}
\bibliographystyle{ieeetr}

\appendix

\end{document}